\documentclass[a4paper,12pt]{article}
\usepackage{enumitem}
\usepackage{calc}

\usepackage{lmodern} 

\usepackage{graphicx} 

\usepackage{amsmath, accents}
\usepackage{amssymb}
\usepackage{amsthm}
\usepackage{amscd}
\usepackage{mathrsfs}

\usepackage{yfonts}

\usepackage{marvosym}

\newtheorem{theorem}			     {Theorem} [section]

\newtheorem{lemma}	      [theorem]  {Lemma}		
\theoremstyle{definition}

\newtheorem{remark} {Remark} 

\newcommand{\C}{\mathbb{C}}
\newcommand{\R}{\mathbb{R}}
\newcommand{\N}{\mathbb{N}}
\newcommand{\Or}{\mathcal{O}}

\newcommand{\K}{\mathbb{K}}

\newcommand{\G}{\Gamma}

\newcommand{\e}{\epsilon}

\newcommand{\Tr}{\textrm{Tr}}

\newcommand{\gue}{\mbox{\textsc{gue}}}
\newcommand{\lue}{\mbox{\textsc{lue}}}

\newcommand{\Ai}{{\rm Ai}}

\usepackage{tikz}
\usetikzlibrary{arrows}
\usetikzlibrary{decorations.pathmorphing}
\usetikzlibrary{decorations.markings}
\usetikzlibrary{patterns}
\usetikzlibrary{automata}
\usetikzlibrary{positioning}
\usepackage{tikz-cd}

\usepackage{pgfplots}

\numberwithin{equation}{section}

\title{Gaussian perturbations of hard edge random matrix ensembles}
\author{Tom Claeys\footnote{Institut de Recherche en Math\'ematique et Physique,  Universit\'e
catholique de Louvain, Chemin du Cyclotron 2, B-1348
Louvain-La-Neuve, BELGIUM}\ \ and Antoine Doeraene\footnotemark[\value{footnote}]}

\tikzset{->-/.style={decoration={
  markings,
  mark=at position #1 with {\arrow{>}}},postaction={decorate}}}
\tikzset{-<-/.style={decoration={
  markings,
  mark=at position #1 with {\arrow{<}}},postaction={decorate}}}

\begin{document}

\maketitle

\begin{abstract}
We study the eigenvalue correlations of random Hermitian $n\times n$ matrices of the form $S=M+\e H$, where $H$ is a GUE matrix, $\e>0$, and $M$ is a positive-definite Hermitian random matrix, independent of $H$, whose eigenvalue density is a polynomial ensemble. 
We show that there is a soft-to-hard edge  transition in the microscopic behaviour of the eigenvalues of $S$ close to $0$ if $\e$ tends to $0$ together with $n\to+\infty$ at a critical speed, depending on the random matrix $M$. In a double scaling limit, we obtain a new family of limiting eigenvalue correlation kernels.
We apply our general results to the cases where (i) $M$ is a Laguerre/Wishart random matrix, (ii) $M=G^*G$ with $G$ a product of Ginibre matrices, (iii) $M=T^*T$ with $T$ a product of truncations of Haar distributed unitary matrices, and (iv) the eigenvalues of $M$ follow a Muttalib-Borodin biorthogonal ensemble.
\end{abstract}

\section{Introduction}

We consider a class of Hermitian random matrices which are perturbed by additive Gaussian noise, and investigate to what extent the microscopic behaviour of the eigenvalues is affected by such a  perturbation, in the limit where the size of the matrices tends to infinity. 
We take the Gaussian noise to be a small multiple of a matrix from the Gaussian Unitary Ensemble (GUE), which consists of Hermitian $n\times n$ matrices $H$ with the probability distribution
\begin{equation}
\label{eq:guedef}
\frac{1}{Z_n^{\gue}} e^{-\frac{n}{2}{\rm Tr}(H^2)} dH,\qquad dH=\prod_{j=1}^ndH_{jj}\,\prod_{1\leq i<j\leq n}d{\rm Re} H_{ij}d{\rm Im} H_{ij}.
\end{equation}
Equivalently, the diagonal entries of 
$H$ are independent identically distributed (iid) normal random variables $\mathcal{N}\left(0, 1/n \right)$ with mean $0$ and variance $1/n$ and the upper triangular entries are iid complex normal random variables  $\mathcal{N}\left(0, \frac{1}{2n} \right) + i\mathcal{N}\left(0, \frac{1}{2n} \right)$. Adding a small multiple of a GUE matrix, $\e H$, to another random matrix $M$ can thus be viewed as an entry-wise Gaussian perturbation of $M$. 
The above normalization of a GUE matrix is such that the large $n$ limit of the mean distribution $\lambda$ of eigenvalues exists and is given by the Wigner semi-circle law
\begin{equation}
\label{eq:semicirclelaw}
d\lambda(x) = \frac{1}{2\pi} \sqrt{4-x^2} dx,\qquad x\in[-2,2].
\end{equation}

Gaussian perturbations of random matrices are closely related to random matrices with external source.
To see this, we note that, given $M$, the probability distribution of the random matrix $S=M+\e H$ can be written as
\[
\frac{1}{Z_n} e^{-\frac{n}{2\e^2}{\rm Tr}((S-M)^2)} dS,\qquad dS=\prod_{j=1}^ndS_{jj}\,\prod_{1\leq i<j\leq n}d{\rm Re} S_{ij}d{\rm Im} S_{ij}.
\]
This is known as the GUE with external source $M$ \cite{BrezinHikami, ZinnJustin}. In our setting, the external source matrix $M$ is not deterministic but is itself a random matrix.

If $M$ is a unitary invariant random matrix with probability measure
\[\frac{1}{Z_n} e^{-n{\rm Tr}\, V(M)} dM,\qquad dM=\prod_{j=1}^ndM_{jj}\,\prod_{1\leq i<j\leq n}d{\rm Re} M_{ij}d{\rm Im} M_{ij},\]
for some potential $V$, then our model is equivalent to a special case of the two-matrix model \cite{BertolaEynardHarnad}, which is defined as a probability measure on pairs of Hermitian matrices $(M_1, M_2)$, given by
\[
\frac{1}{Z_n} e^{-n{\rm Tr}\left(V_1(M_1)+V_2(M_2)-\tau M_1 M_2\right)} dM_1 dM_2,
\]
for certain potentials $V_1, V_2$. If we take \[V_2(x)=x^2/2,\qquad V_1(x)=V(x)+\frac{\tau^2}{2}x^2,\qquad \tau=1/\epsilon,
\]
then it is straightforward to verify that $M_1$ and $M_1-\tau M_2$ are independent, that $M_1$ has the same distribution as $M$, and that $\frac{1}{\tau} M_2$ has the same distribution as the sum $S=M+\epsilon H$, see also \cite[Section 5]{Duits}.

The eigenvalues of Gaussian perturbations of (deterministic or random) matrices can alternatively be realized as the positions of $n$ non-intersecting Brownian paths with a common endpoint and with (deterministic or random) starting points, see e.g.\ \cite{Johansson} and the recent work \cite{ForresterGrela}.

In what follows, the random Hermitian $n\times n$ matrix $M$ has to be independent of the GUE matrix $H$ and such that the joint probability density function of the eigenvalues is of the form
\begin{equation}
\label{eq:pejpdf}
\frac{1}{Z_n} \Delta(x) \det\big[ f_{k-1}(x_j) \big]_{j,k=1}^n,\qquad \Delta(x) = \prod_{1\leq j<k\leq n} (x_k-x_j),\qquad x_1,\ldots, x_n\in\mathbb R,
\end{equation}
for certain functions $f_0,f_1,\ldots, f_{n-1}$, and where $Z_n$ is a normalizing constant.
A density function on $\mathbb R^n$ of this form is called a polynomial ensemble  \cite{KuijlaarsStivigny}. 
For instance, the eigenvalues of unitary invariant random matrix ensembles and of certain products and sums of random matrices follow polynomial ensembles \cite{ClaeysKuijlaarsWang, KieburgKuijlaarsStivigny}.
Polynomial ensembles are special cases of determinantal point processes, their correlation kernel $K_n$ taking the special form
\begin{equation}
\label{eq:pegeneralkernel}
K_n(x,y) = \sum_{j=0}^{n-1} p_j(x)q_j(y),
\end{equation}
where $p_j$ is a polynomial of degree $j$, and $q_j$ is a linear combination of $f_0,\ldots, f_{n-1}$, such that the orthogonality conditions
\begin{equation}
\label{eq:biorthogonality}
\int_{\R} p_j(x) q_k(x) dx = \delta_{jk}, \qquad\mbox{ } j,k = 0,...,n-1,
\end{equation}
are satisfied.
Later on, we will focus on polynomial ensembles defined by functions $f_j$ supported on $[0,+\infty)$, but for now, they can be general.

If the joint eigenvalue density of a random matrix $M$ is a polynomial ensemble with correlation kernel $K_n$, then it was shown in \cite{ClaeysKuijlaarsWang} that the eigenvalues of $S=M+\e H$, with $H$ a GUE matrix independent of $M$ and $\e>0$, also follow a polynomial ensemble, with the transformed eigenvalue correlation kernel
 \begin{equation}
\label{eq:kernelsum}
K_n^S(x,y) = \frac{n}{2\pi i \e^2} \int_{i\R} \int_{\R} K_n(s,t) e^{\frac{n}{2\e^2} ( (x-s)^2 - (y-t)^2)}dtds.
\end{equation}
In addition, if $p_n(x)=\mathbb E[\det(xI-M)]$ is the average characteristic polynomial of $M$, then the average characteristic polynomial of $S=M+\e H$ is given by
\begin{equation}
\label{eq:sumacp}
P_n(x) =\mathbb E[\det(xI-S)]= \frac{\sqrt{n}}{\sqrt{2\pi}i\e} \int_{i\R} p_n(s) e^{\frac{n}{2\e^2} (x-s)^2} ds.
\end{equation}
The formulas \eqref{eq:kernelsum} and \eqref{eq:sumacp} follow from \cite[formulas (2.6) and (2.8)]{ClaeysKuijlaarsWang} after a simple re-scaling argument, and they will be the starting point of our analysis.


\subsection*{Macroscopic eigenvalue behaviour}

The macroscopic large $n$ behaviour of the eigenvalues of $M+\e H$ is well understood thanks to   free probability theory:
if $M=M_n$ is a sequence of random $n\times n$ matrices whose eigenvalue distributions converge almost surely to a measure $\mu$ and if $M$ is independent of the GUE matrix $H=H_n$, then $M$ and $H$ are asymptotically free and we can apply results from free probability theory \cite{Nica-Speicher06, Speicher} to describe the limiting eigenvalue distribution of $S=M+\e H$. Writing
$\lambda_{\e}$ for the rescaled semi-circle law,
\begin{equation}
\label{eq:rescaledsemicircle}
d\lambda_{\e}(x) = \frac{1}{2\pi \e^2} \sqrt{ 4\e^2 - x^2 }dx,\qquad x\in[-2\e,2\e],
\end{equation}
which is the limiting macroscopic density of the eigenvalues of $\e H$, it is well-known that the limiting eigenvalue distribution of $S$ is almost surely given by the free additive convolution $\mu\boxplus\lambda_{\e}$ of $\mu$ and $\lambda_{\e}$, see \cite{Biane} for the definition and properties of the free convolution of a measure $\mu$ with $\lambda_{\e}$.

Another quantity containing global information about random matrix eigenvalues is the zero counting measure of the average characteristic polynomial. The zeros of the average characteristic polynomial are real and simple (see Lemma \ref{lem:zerosavcharpol}), and the zero counting measure can heuristically be interpreted as a typical eigenvalue configuration.
It can therefore be expected that its large $n$ limit coincides with the limiting (mean) eigenvalue distribution in many cases. This is well-known for classical random matrix ensembles and was investigated in a more general framework in \cite{Hardy}. 
The following result about convergence of the zero counting measure of the average characteristic polynomial of $S$ is not surprising in view of the above-mentioned results from free probability. We will prove it in Section \ref{sec:zcmconvergence} directly using the integral representation \eqref{eq:sumacp} and without relying on the more sophisticated results from free probability theory.

\begin{theorem}
\label{thm:zcmconvergence}
Let $M$ be an $n\times n$ Hermitian random matrix such that its eigenvalue density is a polynomial ensemble \eqref{eq:pejpdf}, let $H$ be an $n\times n$ GUE matrix independent of $M$, and let $\e>0$. Write $\mu_n$ for the zero counting measure of the average characteristic polynomial of $M$, and $\nu_n$ for the zero counting measure of the average characteristic polynomial of $S=M+\e H$.

If, for sufficiently large $n$, the support of $\mu_n$ is contained in some $n$-independent compact $K$, and if $\mu_n$ converges weakly to a probability measure $\mu$, then $\nu_n$ converges weakly to $\mu \boxplus \lambda_{\e}$, where $\lambda_{\e}$ is given by \eqref{eq:rescaledsemicircle}.
\end{theorem}

\subsection*{Microscopic eigenvalue behaviour}

From now on, we consider polynomial ensembles supported on $[0,+\infty)$ or on an interval of the form $[0,b]$. We mean by this that the functions $f_j$ in \eqref{eq:pejpdf} are supported on $[0,+\infty)$ or on $[0,b]$. Ensembles of this kind are said to have a \emph{hard edge} at zero.
Classical examples of random matrix ensembles with a hard edge are the Laguerre Unitary Ensemble and the Jacobi Unitary Ensemble.
In these ensembles, the microscopic eigenvalue correlations near $0$ are described in terms of Bessel functions.
As we will see below, other ensembles may lead to other types of microscopic eigenvalue correlations, described in terms of other functions, such as Meijer G-functions or Wright's generalized Bessel functions.
A common feature of all hard edge random matrix ensembles which we will study below, is the existence of a scaling limit near the hard edge of the following form:
\begin{equation}
\label{eq: scaling limit}
\lim_{n \rightarrow +\infty}  \frac{1}{cn^{\gamma}} K_n\left( \frac{u}{cn^{\gamma}}, \frac{v}{cn^{\gamma}} \right) = \K(u,v),\qquad v>0,\ u\in\mathbb C,
\end{equation}
for some values of $c, \gamma>0$, and for some limiting kernel $\K(u,v)$, which depend on the particular choice of random matrix ensemble.

If we consider a Gaussian perturbation of $M$ of the form $S=M+\e H$, even if $\e>0$ is small, the matrix $S$ is typically not positive-definite, in other words the hard edge at $0$ is removed by the perturbation. It is our aim to understand how scaling limits of the eigenvalue correlation kernel near $0$ of the form \eqref{eq: scaling limit} change after the Gaussian perturbation. In particular, we want to see what happens in double scaling limits where the constant $\e$ goes to $0$ as $n$ goes to infinity, as this is the limit in which the soft edge of the spectrum (which we have for fixed $\e>0$) turns into a hard edge at the origin. We now present a general auxiliary result, which we will apply to several concrete examples later on. Given a scaling limit of the form \eqref{eq: scaling limit}, it states that the scaling limit is preserved for the eigenvalue correlation kernel of $S$
provided that $\e\to 0$ sufficiently fast with $n\to\infty$. If $\e\to 0$ at a critical speed, the limiting kernel $\mathbb K$ is deformed.

\begin{lemma}
\label{thm:mainthm}
Consider a sequence of $n\times n$ random matrices $M$ such that their eigenvalue densities are polynomial ensembles on $[0,+\infty)$ or on $[0,b]$. We assume there exist constants $\gamma>1/2$, $c, c_1, c_2, n_0>0$ and $\beta\in [0,1)$ such that the associated correlation kernels $K_n$ satisfy the following conditions:
\begin{enumerate}
\item \label{enum:conv} there exists a function $\K(u,v)$ such that
\begin{equation}
\label{eq:conv2}
\lim_{n \rightarrow +\infty} v^{\beta} \frac{1}{cn^{\gamma}} K_n\left( \frac{u}{cn^{\gamma}}, \frac{v}{cn^{\gamma}} \right) = v^{\beta}\K(u,v),
\end{equation}
uniformly for $u$ in any compact subset of $\mathbb C$ and $v$ in any compact subset of $[0,+\infty)$.
Here  $v^{\beta}\K(u,v)$ and $v^{\beta}K_n\left( \frac{u}{cn^{\gamma}}, \frac{v}{cn^{\gamma}} \right)$ for $v=0$ have to be understood as the limits as $v\to 0$ and it is supposed that these limits exist,
\item \label{enum:growth} for every $(u,v) \in i\R \times [0,+\infty)$ and $n > n_0$,
\begin{equation}
\label{eq:growth2}
|K_n(u,v)| \leq c_1v^{-\beta} n^{\gamma(1-\beta)}  e^{c_2 n^{\gamma} (|u|+|v|)}.
\end{equation}
\end{enumerate}
Let $S=M+\e_n H$ where $H$ is a GUE matrix independent of $M$, and let $K_n^S$ be the eigenvalue correlation kernel for $S$. Then,
\begin{enumerate}
\item if $\e_n$ is such that $\lim_{n\to+\infty}\e_nn^{\gamma-\frac{1}{2}}=0$, we have
\begin{equation}
\label{eq:scaling sum subcrit}
\lim_{n \rightarrow +\infty}  \frac{1}{cn^{\gamma}} K_n^S\left( \frac{x}{cn^{\gamma}}, \frac{y}{cn^{\gamma}} \right) = \K(x,y),\qquad x,y>0,
\end{equation}
\item if $\e_n$ is such that $\lim_{n\to+\infty}c\e_nn^{\gamma-\frac{1}{2}}=\sigma>0$, we have
\begin{equation}
\label{eq:sumconv}
\lim_{n\rightarrow +\infty} \frac{1}{cn^{\gamma}} K_n^S\left( \frac{x}{cn^{\gamma}}, \frac{y}{cn^{\gamma}} \right) = \frac{1}{2\pi i \sigma^2} \int_{i\R} \int_{\R^+} \K(s,t) e^{\frac{1}{2\sigma^2} \left( (x-s)^2 - (y-t)^2\right)} dtds,
\end{equation}
uniformly for $(x,y)$ in compact subsets of $\mathbb C^2$.
\end{enumerate}
\end{lemma}
\begin{remark}
The rate of decay $\Or\left(n^{\frac{1}{2}-\gamma}\right)$ for $\e_n$ appears as a critical speed at which the local eigenvalue behaviour changes. When $\e_n$ goes to $0$ faster than the critical speed, the  eigenvalues of the perturbed random matrix $M+\e_n H$ behave locally near $0$ as if there were no perturbation. At the critical speed, a new limiting kernel appears at $0$, given by (\ref{eq:sumconv}). By a saddle point approximation, it is easy to verify that
\begin{equation}
\lim_{\sigma\to 0}\frac{1}{2\pi i \sigma^2} \int_{i\R} \int_{\R^+} \K(s,t) e^{\frac{1}{2\sigma^2} \left( (x-s)^2 - (y-t)^2\right)} dtds=\mathbb K(x,y),
\end{equation}
which means that \eqref{eq:scaling sum subcrit} and \eqref{eq:sumconv} are consistent.
\end{remark}
\begin{remark}\label{remark: stronger lemma}
Conditions 1 and 2 in the above lemma are designed in such a way that they hold for a large class of random matrix ensembles. 
In some cases, we can just take $\beta=0$.
However, it may happen that the functions $f_j(x)$ defining the polynomial ensemble \eqref{eq:pejpdf} blow up as $x\to 0$. This implies that the kernel $K_n(x,y)$ blows up as $y\to 0$, and thus one cannot expect \eqref{eq:conv2} and \eqref{eq:growth2} to hold for $\beta=0$. This is why we allow $\beta\in[0,1)$.
\end{remark}

In the next section, we discuss several concrete examples of random matrix ensembles to which we can apply Lemma \ref{thm:mainthm}.

\section{Examples}
\subsection{Perturbed Laguerre/Wishart random matrices}
We define the generalized Laguerre Unitary Ensemble (LUE) as the set of $n\times n$ positive-definite Hermitian matrices equipped with the probability measure
\begin{equation}
\label{eq:luedensitygeneral}
\frac{1}{Z_{n,\alpha, k}^{\lue}} (\det M)^{\alpha} e^{-n{\rm Tr} (M^k)} dM,\quad dM=\prod_{j=1}^ndM_{jj}\,\prod_{1\leq i<j\leq n}d{\rm Re} M_{ij}d{\rm Im} M_{ij},\quad  \alpha>-1,\ k\in\N.
\end{equation}
Similarly as for the GUE, the factor $n$ in the exponential ensures the eigenvalues to  remain bounded as $n\to+\infty$ with probability $1$.
For $\alpha\in\mathbb N$ and $k=1$, a random LUE matrix can be realized as $M=G^*G$, where $G$ is a $(n+\alpha)\times n$  complex Ginibre matrix, which has independent identically distributed complex normal entries $\mathcal N(0,\frac{1}{2n})+i\mathcal N(0,\frac{1}{2n})$. In Figure \ref{fig:luenumerics}, we present numerical samples of the perturbed LUE for different values of $\e$.

\begin{figure}
\centering
\begin{tabular}{cccc}
\begin{tikzpicture}[scale = .5]
\begin{axis}[title = {$\e = 2$}]
\addplot[ybar interval, color = blue!70!white, fill = blue!70!white] table [ybar interval] {GUEPEGraphs/LUE_mat4000_e2.plot};
\end{axis}
\end{tikzpicture}
&
\begin{tikzpicture}[scale = .5]
\begin{axis}[title = {$\e = 1/2$}]
\addplot[ybar interval, color = blue!70!white, fill = blue!70!white] table [ybar interval] {GUEPEGraphs/LUE_mat4000_e2-1.plot};
\end{axis}
\end{tikzpicture}
&
\begin{tikzpicture}[scale = .5]
\begin{axis}[title = {$\e = 1/10$}]
\addplot[ybar interval, color = blue!70!white, fill = blue!70!white] table [ybar interval] {GUEPEGraphs/LUE_mat4000_e10-1.plot};
\end{axis}
\end{tikzpicture}
&
\begin{tikzpicture}[scale = .5]
\begin{axis}[title = {$\e = 0$}]
\addplot[ybar interval, color = blue!70!white, fill = blue!70!white] table [ybar interval] {GUEPEGraphs/LUE_mat4000_e0.plot};
\end{axis}
\end{tikzpicture}
\\
\end{tabular}
\caption{Numerical samples of the eigenvalues of $M+\e H$, with $M$ a $4000\times 4000$ matrix drawn from the LUE (with $\alpha = 0$ and $k=1$), and $H$ a $4000\times 4000$ matrix drawn from the GUE. The chosen values for $\e$ are $2$, $1/2$, $1/10$ and $0$. The case $\e=0$ is the unperturbed LUE. The eigenvalues are represented in histograms of $200$ intervals.}
\label{fig:luenumerics}
\end{figure}
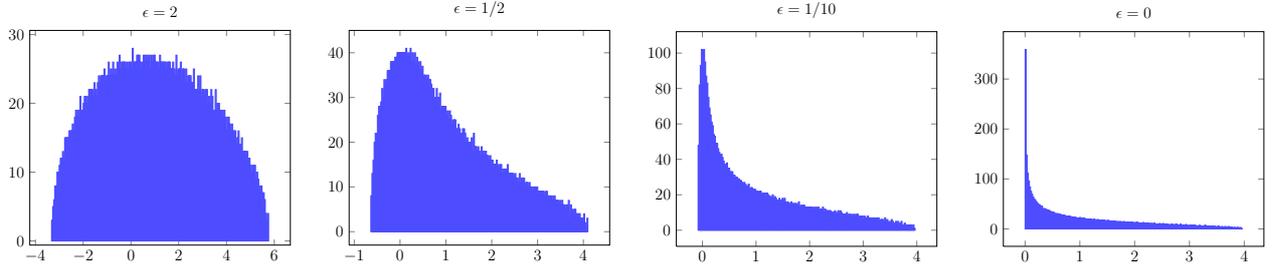

The eigenvalues of a random matrix $M$ with probability distribution \eqref{eq:luedensitygeneral} have the joint probability distribution
\begin{equation}
\label{eq:luejpdf}
\frac{1}{\widetilde{Z}_{n,\alpha,k}^{\lue}} \Delta(x)^2 \prod_{j=1}^n x_j^{\alpha} e^{-nx_j^k} dx_j,\qquad x_1,\ldots, x_n >0,
\end{equation}
which is a polynomial ensemble \eqref{eq:pejpdf} with $f_j(x)=x^{j+\alpha}e^{-nx^k}$ on $\R^+$.
The limiting eigenvalue distribution $\mu$ in this ensemble is almost surely given by a (generalized) Marchenko-Pastur law of the form
\begin{equation}\label{eq:luemacro}
d\mu(x)=\frac{1}{2\pi}\sqrt{\frac{b-x}{x}}h(x) dx,\qquad x\in(0, b),
\end{equation}
for some $n$-independent $b>0$ and polynomial $h$, positive on $[0,b]$. We denote
\begin{equation}\label{def G Stieltjes}
G_\mu(z)=\int\frac{d\mu(x)}{z-x},\qquad z\in\mathbb C\setminus[0,b]
\end{equation}
for the Stieltjes transform of $\mu$.
The limiting eigenvalue distribution of a Gaussian perturbation $S=M+\e H$ is almost surely the free convolution $\mu\boxplus\lambda_\e$.
Its density can be shown to have the form \cite{OlverRao} \begin{equation}\label{eq:LUEmacroSum}
\frac{d(\mu\boxplus\lambda_\e)(x)}{dx}=h_\e(x)\sqrt{(x-a_{\e})(b_{\e}-x)},\qquad x\in[a_{\e}, b_{\e}],
\end{equation} with $h_\e$ positive on $[a_{\e}, b_{\e}]$. The density vanishes like a square root at both edges for any $\e>0$, whereas the density of $\mu$ blows up at the left edge like an inverse square root. 

In \cite{Vanlessen}, large $n$ asymptotics for the eigenvalue correlation kernel $K_n(x,y)$ of $M$ have been obtained using the Deift/Zhou steepest descent method \cite{DeiftZhou} applied to the Riemann-Hilbert problem for generalized Laguerre polynomials. In particular, by \cite[Theorem 2.10 (a)]{Vanlessen}, we have
\begin{equation}
\label{eq:luehardedgescaling}
\lim_{n\to+\infty}\frac{1}{c n^2} K_n\left( \frac{x}{c n^2}, \frac{y}{c n^2} \right) = \K^{\rm Bessel}_{\alpha}(x,y),
\end{equation}
uniformly for $x,y$ in compact subsets of $(0,+\infty)$, with $c=b h(0)^2$. The limiting kernel $\K^{\rm Bessel}_{\alpha}$ is expressed in terms of Bessel functions of the first kind $J_{\alpha}$ and takes the explicit form
\begin{equation}
\label{eq:besselkernel}
\K^{\rm Bessel}_{\alpha}(x,y) = x^{-\frac{\alpha}{2}}y^{\frac{\alpha}{2}} \frac{ J_{\alpha}(\sqrt{x}) \sqrt{y} J_{\alpha}'(\sqrt{y}) - J_{\alpha}(\sqrt{y}) \sqrt{x} J_{\alpha}'(\sqrt{x}) }{ 2 (x-y)}.
\end{equation}
The Bessel kernel is usually defined without the factor $x^{-\frac{\alpha}{2}}y^{\frac{\alpha}{2}}$ in front, such that it is symmetric in $x$ and $y$. The pre-factor is present in our situation because the polynomial ensemble kernel is not symmetric, but it has no effect on the determinants defining the correlation functions associated to the Bessel kernel.
Using Lemma \ref{thm:mainthm} and results from \cite{Vanlessen}, we will prove the following result.

\begin{theorem}
\label{cor:luegeneralresult}
Let $M$ be an $n\times n$ random matrix with probability measure \eqref{eq:luedensitygeneral}, and let $H$ be an $n\times n$ GUE matrix with probability measure \eqref{eq:guedef}, independent of $M$.
Write $K_n^S$ for the eigenvalue correlation kernel of $S=M+\e_n H$.
\begin{itemize}
\item[(i)] {\bf (Sub-critical perturbation)} If $\lim_{n\rightarrow +\infty} \e_n n^{\frac{3}{2}} = 0$, then for $x,y > 0$, we have
\begin{equation}\label{eq:luesubcrit}
\lim_{n \rightarrow +\infty} \frac{1}{cn^2} K_n^S\left( \frac{x}{cn^2}, \frac{y}{cn^2} \right) = \K^{\rm Bessel}_{\alpha}(x,y),
\end{equation}
with $c=b h(0)^2$, where $h$ and $b$ are defined by \eqref{eq:luemacro}.
\item[(ii)] {\bf (Critical perturbation)} If $\lim_{n\rightarrow +\infty} c\e_n n^{\frac{3}{2}} = \sigma>0$, then for $x,y \in \C$, we have
\begin{equation}\label{eq:luecrit}
\lim_{n \rightarrow +\infty} \frac{1}{cn^2} K_n^S\left( \frac{x}{cn^2}, \frac{y}{cn^2} \right) = \frac{1}{2\pi i \sigma^2} \int_{i\R}  \int_{\R^+}  \K^{\rm Bessel}_{\alpha}(s,t)  e^{\frac{1}{2\sigma^2} \left( (s-x)^2 - (t-y)^2 \right)}dtds.
\end{equation}
\item[(iii)] {\bf (Super-critical perturbation)} If $\e_n\to 0$ in such a way that $\e_n n^{\frac{3}{2}}  \to +\infty$ as $n\to+\infty$, then for $x,y \in \C$, we have
\begin{equation}\label{eq:AiryLUE}
\lim_{n \rightarrow +\infty} e^{-\varphi_{\e_n,n}(x) + \varphi_{\e_n,n}(y)}\frac{1}{c_{\e_n} n^{\frac{2}{3}}} K_n^S \left( a_{\e_n} - \frac{x}{c_{\e_n} n^{\frac{2}{3}}}, a_{\e_n} - \frac{y}{c_{\e_n} n^{\frac{2}{3}}} \right) = \K^{\Ai}(x,y),
\end{equation}
with
$a_\e$ as in \eqref{eq:LUEmacroSum}, 
$c_{\e}= \e^{-2}6^{1/3} G_\mu''(u_\e)^{-1/3}$, where $u_\e$ is the unique negative solution of the equation 
\begin{equation}
1+\e^2G_\mu'(u)=0,
\end{equation}
 and with
\[\varphi_{\e,n}(z) = \frac{n^{1/3}z}{c_{\e}\e^2} (u_{\e}-a_{\e}).\]
The limiting kernel $\K^{\Ai}$ is the Airy kernel \begin{equation}\label{eq:Airykernel}
\mathbb K^{\Ai}(u,v) = \frac{\Ai(u) \Ai'(v) - \Ai(v) \Ai'(u)}{u-v}.
\end{equation}
\end{itemize}
\end{theorem}
\begin{remark}
The natural interpretation of these results is as follows: if $\e_n$ tends to $0$ sufficiently fast, then the perturbation is too weak to have an effect on the large $n$ behaviour of the eigenvalues near $0$. In this case we have the same Bessel kernel limit as for the unperturbed LUE, even though $0$ is not a hard edge any longer (for any $\epsilon>0$ and $n$ fixed, there is a non-zero probability of having negative eigenvalues).
On the other hand, if $\e_n$ tends to zero slowly, one is close to the fixed $\e$ case where one has, macroscopically, soft edges, which suggests Airy behaviour. 
In \eqref{eq:AiryLUE}, one should note that $c_{\e}$ blows up as $\e\to 0$: it is of the order $\e^{-8/9}$.
The intuition behind this, is that the typical distance between eigenvalues near $a_\e$ is of the order $\e^{8/9}n^{-2/3}$.
If $\e_n\to 0$ at the critical speed, the typical distance between eigenvalues is of the order $n^{-2}$ and it is on this scale that the actual transition between the Bessel and the Airy kernel takes place.
\end{remark}
\begin{remark}
One could consider more general LUE type ensembles where the monomial $M^k$ in \eqref{eq:luedensitygeneral} is replaced by a polynomial $V(M)$. As long as $V$ is such that the limiting eigenvalue density blows up like an inverse square root near $0$, it leads no doubt that similar scaling limits can be obtained, but we do not investigate this further.
\end{remark}
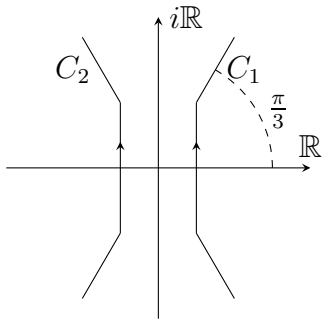
\begin{figure}
\centering
\begin{tikzpicture}[> = stealth]

\draw[thin, ->] (-2, 0) -- (2,0);
\draw (2,0) node [anchor = south] {$\R$};
\draw[thin, ->] (0, -2) -- (0,2);
\draw (0,2) node [anchor = west] {$i\R$};

\draw (60:1cm) -- node [anchor = west] {$C_1$} (60:2cm);
\draw[->- = .7] (300:1cm) -- (60:1cm);
\draw (300:2cm) -- (300:1cm);

\draw (120:1cm) -- node [anchor = east] {$C_2$} (120:2cm);
\draw[->- = .7] (240:1cm) -- (120:1cm);
\draw (240:2cm) -- (240:1cm);


\draw[dashed, thin] (1.5,0) arc (0:60:1.5cm);
\draw (30:1.5cm)  node[anchor = west] {$\frac{\pi}{3}$};
\end{tikzpicture}
\caption{Contours in the definition of the Airy kernel $\K^{\Ai}$}
\label{fig:airykernelcontours}
\end{figure}
Parts (i) and (ii) of Theorem \ref{cor:luegeneralresult} will be direct consequences of Lemma \ref{thm:mainthm}. We need the results from \cite{Vanlessen} to show that the conditions of Lemma \ref{thm:mainthm} are fulfilled, and also to prove part (iii). Here, instead of \eqref{eq:Airykernel}, we will arrive at a different representation of the Airy kernel:
\begin{equation}
\label{eq:airykernel}
\mathbb K^{\Ai}(u,v) = \frac{1}{4 \pi^2}\int_{C_2} ds \int_{C_1} dt \frac{1}{s-t} \frac{e^{t^3-vt}}{e^{s^3-us}},
\end{equation}
with $C_1$ going from $e^{-\frac{i\pi}{3}}\infty$ to $e^{\frac{i\pi}{3}}\infty$ and $C_2$ its reflexion through the vertical axis, as in Figure \ref{fig:airykernelcontours}. 
Using integration by parts on the integral representation of the Airy function \cite[Formula 9.5.4]{NIST}, one can easily check that both kernels \eqref{eq:Airykernel} and \eqref{eq:airykernel} are indeed the same.

\subsection{Perturbed products of Ginibre matrices}
\label{sec:prodginibre}

Products of Ginibre matrices have been studied intensively during the last years, see e.g.\ \cite{Akemann-Ipsen15, Akemann-Ipsen-Kieburg13, Akemann-Kieburg-Wei13,  Forrester, KuijlaarsStivigny, KuijlaarsZhang}.
The squared singular values of products of $m > 0$ independent complex Ginibre matrices follow also a polynomial ensemble with a hard edge at $0$. Let $Y_m = X_mX_{m-1} ... X_1$, with $X_j$  an $(n + \nu_j) \times (n + \nu_{j-1})$ matrix with complex standard Gaussian iid entries, and with the $X_j$'s independent. The $\nu_j$'s are assumed to be non-negative integers, and $\nu_0 = 0$. For $n$ fixed, it was proved in \cite{KuijlaarsStivigny} that the joint density of the squared singular values of $Y_m$ is a polynomial ensemble. For $m=1$, we recover the LUE with $k=1$ and $\alpha = \nu_1$ after rescaling. Numerical samples for perturbed products of Ginibre matrices are presented in Figure \ref{fig:ginibrenumerics}.

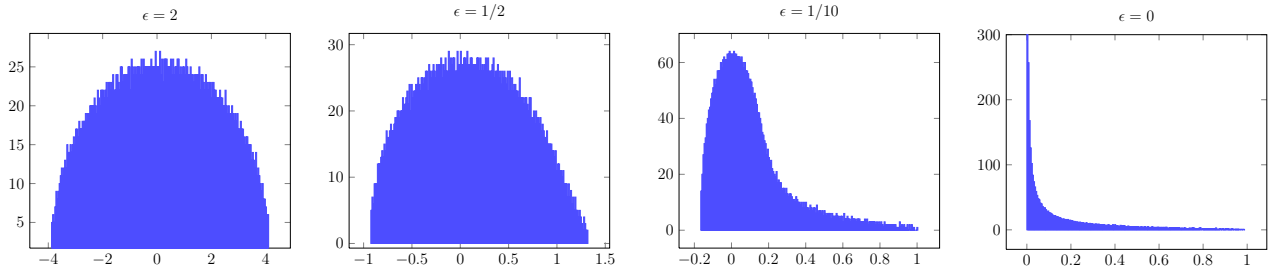
\begin{figure}
\centering
\begin{tabular}{cccc}
\begin{tikzpicture}[scale = .5]
\begin{axis}[title = {$\e = 2$}]
\addplot[ybar interval, color = blue!70!white, fill = blue!70!white] table [ybar interval] {GUEPEGraphs/Prod_Gin_mat4000_e2.plot};
\end{axis}
\end{tikzpicture}
&
\begin{tikzpicture}[scale = .5]
\begin{axis}[title = {$\e = 1/2$}]
\addplot[ybar interval, color = blue!70!white, fill = blue!70!white] table [ybar interval] {GUEPEGraphs/Prod_Gin_mat4000_e2-1.plot};
\end{axis}
\end{tikzpicture}
&
\begin{tikzpicture}[scale = .5]
\begin{axis}[title = {$\e = 1/10$}]
\addplot[ybar interval, color = blue!70!white, fill = blue!70!white] table [ybar interval] {GUEPEGraphs/Prod_Gin_mat4000_e10-1.plot};
\end{axis}
\end{tikzpicture}
&
\begin{tikzpicture}[scale = .5]
\begin{axis}[
	title = {$\e = 0$},
	ymax = 300
]
\addplot[ybar interval, color = blue!70!white, fill = blue!70!white] table [ybar interval] {GUEPEGraphs/Prod_Gin_mat_e0.plot};
\end{axis}
\end{tikzpicture}
\\
\end{tabular}
\caption{Numerical samples of the eigenvalues  of $\frac{3^3}{4^4 4000^3}Y_3^*Y_3 + \e H$, where $Y_3$ is a product of three $4000 \times 4000$ (i.e., $\nu_j = 0$ for $j=1,2,3$) Ginibre matrices, and $H$ is a $4000 \times 4000$ matrix drawn from the GUE. Because of the factor in front of $Y_3$, the eigenvalues accumulate on $[0,1]$. The chosen values for $\e$ are $2$, $1/2$, $1/10$ and $0$. The eigenvalues are represented in  histograms of $200$ intervals. The first column for the $\e=0$ case has been truncated to $300$ eigenvalues, its actual value being around $1500$.}
\label{fig:ginibrenumerics}
\end{figure}

The correlation kernel for the squared singular values of $Y_m$, or the eigenvalues of $Y_m^*Y_m$, is given by \cite[Formula (5.1)]{KuijlaarsZhang}
\begin{equation}
\label{eq:prodginibrekernel}
K_n(x,y) = \frac{1}{(2\pi i)^2} \int_{-\frac{1}{2}+i\R} ds \int_{\Sigma_n} dt \prod_{j=0}^m \frac{\G(s+\nu_j+1)}{\G(t+\nu_j+1)} \frac{\G(t-n+1)}{\G(s-n+1)} \frac{x^t y^{-s-1}}{s-t},
\end{equation}
where $\Gamma$ denotes the Euler Gamma-function and
where $\Sigma_n$ is a closed contour encircling $0,1,...,n$ in the positive direction in such a way that ${\rm Re}\, t > -\frac{1}{2}$ for $t\in \Sigma_n$. The largest eigenvalue of this ensemble is typically of order $n^m$ \cite{PensonZyckowski}, and it is therefore more natural for us to rescale the kernel in the following way
\begin{multline}\label{eq:prodginibrekernelrescaled}
\widetilde{K}_n(x,y) := n^m K_n \left( n^mx, n^my \right) \\
= \frac{1}{(2\pi i)^2} \int_{-\frac{1}{2}+i\R} ds \int_{\Sigma_n} dt \prod_{j=0}^m \frac{\G(s+\nu_j+1)}{\G(t+\nu_j+1)} \frac{\G(t-n+1)}{\G(s-n+1)} \frac{x^t y^{-s-1} n^{m(t-s)}}{s-t}.
\end{multline}
This is the correlation kernel for the eigenvalues of $\frac{1}{n^m}Y_m^*Y_m$.
Using this normalization, it has been shown \cite[Theorem 3.2]{Neuschel} that the zero counting measures of the average characteristic polynomials converge (in the weak-$*$ sense) to the Fuss-Catalan distribution \cite{PensonZyckowski}. We may apply Theorem \ref{thm:zcmconvergence}, and this implies that the counting measures of the average characteristic polynomials of
 the perturbed random matrix $\frac{1}{n^m}Y_m^*Y_m+\e H$ converges to the free additive convolution of the Fuss-Catalan distribution with the semi-circle law $\lambda_\e$ for any $\epsilon>0$.

The microscopic behaviour of the eigenvalues near the origin is described by the following scaling limit: we have \cite{KuijlaarsZhang}
\begin{equation}
\label{eq:prodginibrehardedgescaling}
\lim_{n\rightarrow +\infty} \frac{1}{n^{m+1}} \widetilde{K}_n\left( \frac{x}{n^{m+1}}, \frac{y}{n^{m+1}} \right) = \K_{\nu}^{\rm G}(x,y),
\end{equation}
for $x , y > 0$, where
\begin{equation}
\label{eq:prodginibrelimkernel}
\K_{\nu}^{\rm G}(x,y) = \frac{1}{(2\pi i)^2} \int_{-\frac{1}{2}+i\R} ds \int_{\Sigma} dt \prod_{j=0}^m \frac{\G(s+\nu_j+1)}{\G(t+\nu_j+1)} \frac{\sin \pi s}{\sin \pi t} \frac{x^t y^{-s-1}}{s-t}.
\end{equation}
The contour $\Sigma$ comes from $+\infty$ in the upper half plane, encircles the positive real axis and goes back to $+\infty$ in the lower half plane, in such a way that ${\rm Re}\, t > -\frac{1}{2}$ for $t\in \Sigma$. The kernel $\K_\nu^{\rm G}$ can also be expressed in terms of Meijer $G$-functions. Recently, sine and Airy kernel limits were confirmed rigorously in the bulk and at the right edge \cite{LiuWangZhang}.

Using Lemma \ref{thm:mainthm}, we will prove the following result in Section \ref{sec:prodginibreproofs}.

\begin{theorem}\label{thm: ginibre}
Let $Y_m = X_mX_{m-1} ... X_1$, with the $X_j$'s independent $(n + \nu_j) \times (n + \nu_{j-1})$ complex Ginibre matrices, $\nu_0=0$, and let $H$ be an $n\times n$ GUE matrix independent of $Y_m$.
Write  $K_n^S$ for the eigenvalue correlation kernel of $S=\frac{1}{n^m}Y_m^*Y_m+\e_n H$.
\begin{itemize}
\item[(i)] {\bf (Sub-critical perturbation)} If $\lim_{n\rightarrow +\infty} \e_n n^{m+\frac{1}{2}} = 0$, then for $x,y > 0$, we have
\begin{equation}\label{ginibresubcrit}
\lim_{n \rightarrow +\infty} \frac{1}{n^{m+1}} K_n^S\left( \frac{x}{n^{m+1}}, \frac{y}{n^{m+1}} \right) = \K_{\nu}^{\rm G}(x,y).
\end{equation}
\item[(ii)] {\bf (Critical perturbation)} If $\lim_{n\rightarrow +\infty} \e_n n^{m+\frac{1}{2}} = \sigma>0$, then for $x,y \in \C$, we have
\begin{equation}\label{ginibrecrit}
\lim_{n \rightarrow +\infty} \frac{1}{n^{m+1}} K_n^S\left( \frac{x}{n^{m+1}}, \frac{y}{n^{m+1}} \right) = \frac{1}{2\pi i \sigma^2} \int_{i\R}  \int_{\R^+}  \K^{\rm G}_{\nu}(s,t)  e^{\frac{1}{2\sigma^2} \left( (s-x)^2 - (t-y)^2 \right)} dtds.
\end{equation}
\end{itemize}
\end{theorem}

\begin{remark}
In the super-critical regime, one expects Airy behaviour just like in the LUE case.
To prove this, one could try to follow the same steps as for the LUE, but this will become considerably harder because the kernel $K_n$ is now expressed in terms of multiple orthogonal polynomials instead of (generalized) Laguerre polynomials. 
We will come back to this issue later on in Remark \ref{remark: Airy Ginibre}.
\end{remark}

\begin{remark}
The limiting kernel (\ref{eq:prodginibrelimkernel}), in the case $m=2$, appears also at the hard edge of a matrix from the Cauchy-Laguerre two-matrix model \cite{BGS}. This is the space of pairs of positive-definite Hermitian  $n\times n$ matrices with probability measure
\begin{equation}
\frac{1}{Z_n} \frac{ \det(M_1)^{a}\det(M_2)^b e^{-\Tr(c_1M_1+c_2M_2)}}{\det(M_1+M_2)^n} dM_1dM_2,
\end{equation}
with $a,b > -1$, $a+b > -1$ and $c_1, c_2>0$. The correlation kernel for the eigenvalues of one of the matrices is given as a double contour integral similar to \eqref{eq:prodginibrekernel}, and we expect that Lemma \ref{thm:mainthm} can be applied to this case as well.
\end{remark}
\begin{remark}
A different type of (deterministic) perturbation of products of Ginibre matrices has been studied in \cite{ForresterLiu}.
\end{remark}

\subsection{Perturbed products of truncated unitary matrices}
\label{sec:truncunitary}

Another example is given by the squared singular values of products of $m>0$ truncated unitary matrices. As for the previous case, we form $Y_m = T_m ... T_1$, but now $T_j$ is the upper left $(n+\nu_j)\times(n+\nu_{j-1})$ truncation of a unitary matrix of size $\ell_j \times \ell_j$ drawn randomly from the unitary group $\mathcal{U}(\ell_j)$ equipped with the Haar measure, as in \cite{Zyczkowski-Sommers00}. We assume that $\nu_0 = 0$, that $\nu_1, ..., \nu_m$ are non-negative integers, and that $\ell_j \geq n+\nu_j+1$ for $j = 1, ..., m$. See Figure \ref{fig:truncnumerics} for numerical realizations of this ensemble perturbed by  additive Gaussian noise.
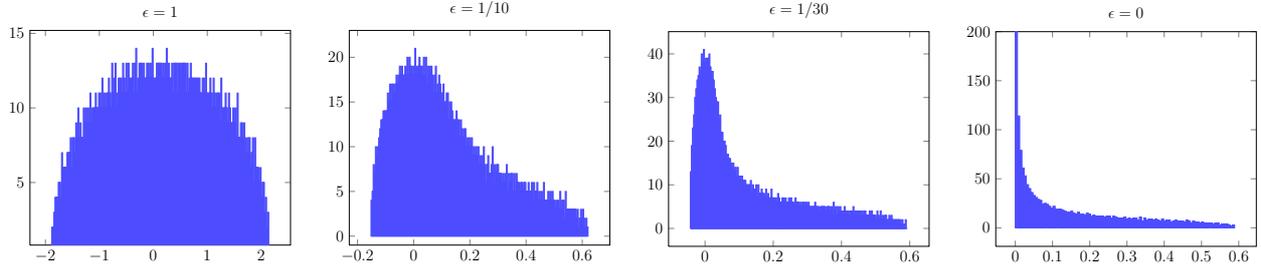
\begin{figure}
\centering
\begin{tabular}{cccc}
\begin{tikzpicture}[scale = .5]
\begin{axis}[title = {$\e = 1$}]
\addplot[ybar interval, color = blue!70!white, fill = blue!70!white] table [ybar interval] {GUEPEGraphs/Prod_Trunc_mat2000_e1.plot};
\end{axis}
\end{tikzpicture}
&
\begin{tikzpicture}[scale = .5]
\begin{axis}[title = {$\e = 1/10$}]
\addplot[ybar interval, color = blue!70!white, fill = blue!70!white] table [ybar interval] {GUEPEGraphs/Prod_Trunc_mat2000_e10-1.plot};
\end{axis}
\end{tikzpicture}
&
\begin{tikzpicture}[scale = .5]
\begin{axis}[title = {$\e = 1/30$}]
\addplot[ybar interval, color = blue!70!white, fill = blue!70!white] table [ybar interval] {GUEPEGraphs/Prod_Trunc_mat2000_e30-1.plot};
\end{axis}
\end{tikzpicture}
&
\begin{tikzpicture}[scale = .5]
\begin{axis}[
	title = {$\e = 0$},
	ymax = 200,
]
\addplot[ybar interval, color = blue!70!white, fill = blue!70!white] table [ybar interval] {GUEPEGraphs/Prod_Trunc_mat2000_e0.plot};
\end{axis}
\end{tikzpicture}
\\
\end{tabular}
\caption{Numerical distribution of the eigenvalues of $Y_3^*Y_3 + \e H$, for $H$ a $2000\times 2000$ GUE matrix and $Y_3$ is the product of three $2000\times 2000$ truncations of three $4005\times 4005$ Haar matrices. The chosen values for $\e$ are $1$, $1/10$, $1/30$ and 0. The eigenvalues are represented in a histogram of 200 intervals. The first column for the $\e=0$ case has been truncated to 200, the actual value being around 550.}
\label{fig:truncnumerics}
\end{figure}

If $\ell_1 \geq 2n+\nu_1$, it was shown in \cite{KieburgKuijlaarsStivigny} that the joint probability density of the squared singular values is a polynomial ensemble, whose kernel is given by
\begin{equation}
\label{eq:KieburgKuijlaarsStivignykernel}
K_n(x,y) = \frac{1}{(2\pi i)^2} \int_C ds \int_{\Sigma_n} dt \prod_{j=0}^m \frac{ \Gamma(s+1+\nu_j) \Gamma(t+1+\ell_j-n) }{ \Gamma(t+1+\nu_j) \Gamma(s+1+\ell_j-n) } \frac{ x^t y^{-s-1} }{s-t}.
\end{equation}
The contour $C$ leaves at $-\infty$ in the lower half plane, encircles the semi axis $(-\infty, -1)$ and returns to $-\infty$ in the positive half plane, $\Sigma_n$ being the same contour as in Section \ref{sec:prodginibre}. Moreover, the contours $C$ and $\Sigma_n$ are not allowed to intersect. In \cite{ClaeysKuijlaarsWang}, \eqref{eq:KieburgKuijlaarsStivignykernel} was proved under the weaker assumption $\sum_{j=1}^m(\ell_j-n-\nu_j)\geq n$, instead of $\ell_1 \geq 2n+\nu_1$.

This kernel also has a limiting kernel appearing near the hard edge. As $n$ goes to infinity, we also have to let $\ell_1, ..., \ell_m$ go to infinity. For each $\ell_j$, we may choose either to let $\ell_j-n$ go to infinity, or to keep $\ell_j-n$ fixed. For each of these choices, the scaling leads to a different limiting kernel. We thus take $J \subseteq \{2,...,m\}$ a subset of indices. We then let $\ell_1,...,\ell_m$ go to infinity in such a way that
\begin{align}
& \ell_k - n \rightarrow +\infty, &\mbox{ if } k \notin J,\\
&\ell_k - n = \mu_k, &\mbox{ if }k = j_k \in J.
\end{align}
Define finally $c_n = n\prod_{j\notin J} (\ell_j-n)$. The kernel  (\ref{eq:KieburgKuijlaarsStivignykernel}) then has the following scaling limit \cite[Theorem 2.8]{KieburgKuijlaarsStivigny} for $x\in\mathbb C, y>0$,
\begin{multline}
\label{eq:KieburgKuijlaarsStivignylimkernel}
\lim_{n\rightarrow +\infty} \frac{1}{c_n} K_n\left( \frac{x}{c_n}, \frac{y}{c_n} \right) = \K_{\nu,\mu}^{\rm T}(x,y)\\
:=\frac{1}{(2\pi i)^2} \int_{-\frac{1}{2}+i\R} ds \int_{\Sigma} dt \prod_{j=0}^m \frac{\Gamma( s+1+\nu_j )}{\Gamma( t+1+\nu_j )} \frac{\sin \pi s}{\sin \pi t} \prod_{k\in J} \frac{\Gamma( t+1+\mu_k )}{\Gamma( s+1+\mu_k )} \frac{x^t y^{-s-1}}{s-t}.
\end{multline}
The contour $\Sigma$ is the same as in Section \ref{sec:prodginibre}. 
Note that if $J$ is empty, the limiting kernel (\ref{eq:KieburgKuijlaarsStivignylimkernel}) reduces to the kernel (\ref{eq:prodginibrelimkernel}). 
As eigenvalues of a product of truncated unitary matrices, the eigenvalues of $M$ remain bounded as $n\to+\infty$. It can be verified, in a similar way as we will do in the case of products of Ginibre matrices, that the eigenvalue correlation kernel for $M$ satisfies conditions similar to those of Lemma \ref{thm:mainthm}, if we replace $cn^\gamma$ by $c_n$ (see Remark \ref{remark: analogy}). This will allow us to prove the following.
\begin{theorem}\label{thm: trunc}
Let $Y_m$ be a product of truncations of unitary Haar distributed matrices as described above, such that the eigenvalue correlation kernel of  $Y_m^*Y_m$ is given by \eqref{eq:KieburgKuijlaarsStivignykernel}, and let $H$ be an $n\times n$ GUE matrix independent of $Y_m$. Write  $K_n^S$ for the eigenvalue correlation kernel of $S=Y_m^*Y_m+\e_n H$. Then, we have
\begin{itemize}
\item[(i)] {\bf (Sub-critical perturbation)} If $\lim_{n\rightarrow +\infty} \e_n c_n n^{-\frac{1}{2}} = 0$, then for $x,y > 0$, we have
\begin{equation}\label{truncsubcrit}
\lim_{n \rightarrow +\infty} \frac{1}{c_n} K_n^S\left( \frac{x}{c_n}, \frac{y}{c_n} \right) = \K_{\nu,\mu}^{\rm T}(x,y).
\end{equation}
\item[(ii)] {\bf (Critical perturbation)} If $\lim_{n\rightarrow +\infty} \e_n c_n n^{-\frac{1}{2}} = \sigma>0$, then for $x,y \in \C$, we have
\begin{equation}\label{trunccrit}
\lim_{n \rightarrow +\infty} \frac{1}{c_n} K_n^S\left( \frac{x}{c_n}, \frac{y}{c_n} \right) = \frac{1}{2\pi i \sigma^2} \int_{i\R}  \int_{\R^+}  \K_{\nu,\mu}^{\rm T}(s,t)  e^{\frac{1}{2\sigma^2} \left( (s-x)^2 - (t-y)^2 \right)} dtds.
\end{equation}
\end{itemize}
\end{theorem}

\begin{remark}
If $m=1$, we have a single truncation of a unitary Haar matrix. Its squared singular values are in the Jacobi Unitary Ensemble, and in this case the kernel $\K_{\nu,\mu}^{\rm T}$ reduces to the Bessel kernel $\K_{\nu}^{\rm Bessel}$.
\end{remark}

\subsection{Perturbed Muttalib-Borodin biorthogonal ensembles}

The last example consists of random matrices for which the joint probability density of eigenvalues is the Muttalib-Borodin Laguerre ensemble \cite{Borodin, Muttalib}
\begin{equation}\label{MB}
\frac{1}{Z_n}\prod_{j < k} (\lambda_k - \lambda_j) (\lambda_k^{\theta} - \lambda_j^{\theta}) \prod_{j=1}^n \lambda_j^\alpha e^{-n\lambda_j},\qquad \theta>0,\alpha>-1.
\end{equation}
Such densities can be realized as eigenvalue densities of random matrices, see \cite {AdlervanMoerbekeWang, Cheliotis, ForresterWang}.
In \cite{ForresterWang}, the authors constructed a random matrix with this eigenvalue density in the following way, in the case where $\theta$ is a positive integer and $\alpha$ a non-negative integer. Define $\alpha_j$, $j = 1, ..., n$ by $\alpha_j = \theta(j-1)+\alpha$. Then, consider the matrix $X$ of size $m \times n$, with $m \geq n+(n-1)\theta+\alpha$, whose $(j,k)$ entry is $0$ if $j-k >\alpha_k$, and following independent standard complex Gaussian distributions otherwise. In other words, $X$ is a complex Ginibre matrix, but with the entries in a certain region in the lower left corner of the matrix replaced by zeros. The eigenvalues of such matrices are shown for $\theta = 3$ and $\alpha = 1$ in Figure \ref{fig:muttaborodnumerics}.
\begin{figure}
\centering
\begin{tabular}{cccc}
\begin{tikzpicture}[scale = .5]
\begin{axis}[title = {$\e = 2$}]
\addplot[ybar interval, color = blue!70!white, fill = blue!70!white] table [ybar interval] {GUEPEGraphs/MuttaBorod_mat1500_e2.plot};
\end{axis}
\end{tikzpicture}
&
\begin{tikzpicture}[scale = .5]
\begin{axis}[title = {$\e = 1/2$}]
\addplot[ybar interval, color = blue!70!white, fill = blue!70!white] table [ybar interval] {GUEPEGraphs/MuttaBorod_mat1500_e2-1.plot};
\end{axis}
\end{tikzpicture}
&
\begin{tikzpicture}[scale = .5]
\begin{axis}[title = {$\e = 1/10$}]
\addplot[ybar interval, color = blue!70!white, fill = blue!70!white] table [ybar interval] {GUEPEGraphs/MuttaBorod_mat1500_e10-1.plot};
\end{axis}
\end{tikzpicture}
&
\begin{tikzpicture}[scale = .5]
\begin{axis}[title = {$\e = 0$}]
\addplot[ybar interval, color = blue!70!white, fill = blue!70!white] table [ybar interval] {GUEPEGraphs/MuttaBorod_mat1500_e0.plot};
\end{axis}
\end{tikzpicture}
\\
\end{tabular}
\caption{Numerical samples of eigenvalues of $\frac{1}{1500} X^*X + \e H$, for $H$ a GUE matrix and $X$ a $1500\times 1500$ matrix as described below (\ref{MB}), with $\theta = 3$ and $\alpha = 1$. Values for $\e$ have been taken as $2$, $1/2$, $1/10$ and $0$. The eigenvalues are represented in  histograms with 100 intervals.}
\label{fig:muttaborodnumerics}
\end{figure}
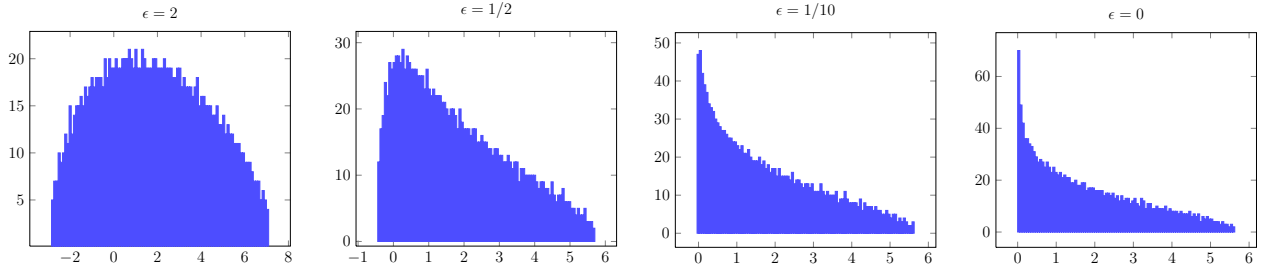
The density of the eigenvalues of $\frac{1}{n} X^*X$ is then given by \eqref{MB}.
The eigenvalue correlation kernel of $\frac{1}{n}X^*X$ can be expressed as \cite{ForresterWang}
\begin{equation}\label{kernel MB}
K_n(x,y) = \frac{1}{(2\pi i)^2} e^{n(x-y)}\int_{C} ds \int_{C_{\alpha}} dt \frac{x^{-s-1}y^t}{s-t} n^{-s+t} \frac{\Gamma(s+1)}{\Gamma(t+1)} \prod_{k=1}^n \frac{s-\alpha_k}{t-\alpha_k},
\end{equation}
with $C_{\alpha}$ a contour enclosing $\alpha_1,...,\alpha_n$ and $C$ starting at $-\infty$ in the lower half plane, enclosing $C_{\alpha}$ and going back to $-\infty$ in the upper half plane. An alternative expression was given in \cite{Zhang}:
\begin{equation}\label{kernel MB2}
K_n(x,y) = \frac{\theta}{(2\pi i)^2} \int_{c+i\mathbb R} ds \int_{\Sigma_n} dt \frac{x^{-\theta s-1}y^{\theta t}}{s-t} n^{-\theta s+\theta t} \frac{\Gamma(s+1)\Gamma(\alpha+1+\theta s)\Gamma(t-n+1)}{\Gamma(t+1)\Gamma(\alpha+1+\theta t)\Gamma(s-n+1)} ,
\end{equation}
with $c=-\frac{1}{2}+\frac{1}{2}\max\{0,1-\frac{\alpha+1}{\theta}\}$, and $\Sigma$ a closed counter-clockwise contour going around $0,1,\ldots, n-1$ and for which ${\rm Re }\, t>c$.

It admits the hard edge scaling limit \cite{Borodin, ForresterWang}
\begin{equation}
\lim_{n \rightarrow +\infty} \frac{1}{n^{1+\frac{1}{\theta}}} K_n\left(\frac{x}{n^{1+\frac{1}{\theta}}}, \frac{y}{n^{1+\frac{1}{\theta}}}\right) = \mathbb K_{\alpha,\theta}^{\textrm{MB}}(x,y),\qquad x\in\mathbb C,\ y>0,\label{scaling limit MB}
\end{equation}
where the limiting kernel $\mathbb K_{\alpha, \theta}^{\textrm{MB}}$ is given by
\begin{equation}\label{lim kernel MB}
\mathbb K_{\alpha,\theta}^{\textrm{MB}}(x,y) := \frac{\theta}{(2\pi i)^2} \left( \frac{y}{x} \right)^{\alpha} \int_{C^{\delta}} ds \int_{\Sigma} dt \frac{x^{-\theta s -1} y^{\theta t}}{s-t} \frac{\Gamma(\theta s+\alpha+1)}{\Gamma(\theta t+\alpha+1)} \frac{\Gamma(s+1)}{\Gamma(t+1)} \frac{\sin \pi s}{\sin \pi t},
\end{equation}
with $\Sigma$ starting from $+\infty$ in the upper half plane, enclosing the positive real axis and going back to $+\infty$ in the lower half plane, and $C^{\delta}$ consisting of two rays starting from $-\frac{1}{2}$ and making an angle $0 < \delta < \frac{\pi}{2}$ with the vertical axis, oriented upwards.

Similarly as in the previous cases, we have the following result.
\begin{theorem}\label{thm:MB}
Let $M$ be a random matrix with eigenvalue density \eqref{MB} and correlation kernel \eqref{kernel MB}, and let $H$ be an $n\times n$ GUE matrix independent of $X$.
Write  $K_n^S$ for the eigenvalue correlation kernel of $S=M+\e_n H$.
\begin{itemize}
\item[(i)] {\bf (Sub-critical perturbation)} If $\lim_{n\rightarrow +\infty} \e_n n^{\frac{1}{\theta}+\frac{1}{2}} = 0$, then for $x,y > 0$, we have
\begin{equation}
\lim_{n \rightarrow +\infty} \frac{1}{n^{\frac{1}{\theta}+1}} K_n^S\left( \frac{x}{n^{\frac{1}{\theta}+1}}, \frac{y}{n^{\frac{1}{\theta}+1}} \right) = \K_{\alpha,\theta}^{{\rm MB}}(x,y).
\end{equation}
\item[(ii)] {\bf (Critical perturbation)} If $\lim_{n\rightarrow +\infty} \e_n n^{\frac{1}{\theta}+\frac{1}{2}} = \sigma>0$, then for $x,y \in \C$, we have
\begin{equation}
\lim_{n \rightarrow +\infty} \frac{1}{n^{\frac{1}{\theta}+1}} K_n^S\left( \frac{x}{n^{\frac{1}{\theta}+1}}, \frac{y}{n^{\frac{1}{\theta}+1}} \right) = \frac{1}{2\pi i \sigma^2} \int_{i\R}  \int_{\R^+}  \K_{\alpha,\theta}^{{\rm MB}}(s,t)  e^{\frac{1}{2\sigma^2} \left( (s-x)^2 - (t-y)^2 \right)}dtds.
\end{equation}
\end{itemize}
\end{theorem}

\begin{remark}
The hard edge scaling limit \eqref{scaling limit MB} was derived in \cite{Borodin} with a different expression for the limiting kernel,
\begin{equation}
\mathbb{K}_{\alpha,\theta}^{\textrm{MB}}(x,y) = \theta y ^\alpha \int_0^1 J_{(\alpha+1)/\theta, 1/\theta}(xu) J_{\alpha+1,\theta}(yu)^{\theta} u^\alpha du,
\end{equation}
where $J_{a,b}$ is Wright's generalization of Bessel functions given by
\begin{equation}
J_{a,b}(x) = \sum_{j=0}^{+\infty} \frac{ (-x)^j }{ j! \Gamma(a+jb) }.
\end{equation}
If $1/\theta\in\mathbb N$, this limiting kernel can be expressed in terms of the  kernel $\mathbb K_\nu^{\rm G}$ appearing for products of Ginibre matrices, see \cite[Section 5]{KuijlaarsStivigny}. 
\end{remark}

\bigskip

Yet another hard edge limiting kernel was obtained in \cite{AtkinClaeysMezzadri, XDZ} in random matrix ensembles with singularities of the form
\begin{equation}
\frac{1}{Z_n}e^{-n\Tr(M+t^k/M^k)}dM,\qquad k\in\mathbb N,
\end{equation}
on the space of $n\times n$ positive-definite Hermitian matrices. A limiting kernel was obtained which can be expressed in terms of the Painlev\'e III hierarchy, the corresponding value of $\gamma$ in \eqref{eq: scaling limit} is $\gamma=2$. We believe that Lemma \ref{thm:mainthm} can also be applied to this ensemble, but a detailed study would lead us too far.

\subsection*{Outline}

In Section \ref{sec:zcmconvergence}, we prove Theorem \ref{thm:zcmconvergence} on the convergence of the zero counting measures of the average characteristic polynomials.
In Section \ref{sec:proofmainthm}, we prove the central auxiliary result of this paper, Lemma \ref{thm:mainthm}. Parts (i) and (ii) of Theorem \ref{cor:luegeneralresult} on perturbed LUE matrices are proved in Section \ref{sec:lueproofs}, and part (iii) in Section \ref{sec:lueproofs2}. In Section \ref{sec:prodginibreproofs}, the proof of Theorem \ref{thm: ginibre} on perturbed Ginibre products is given. The proofs of Theorem \ref{thm: trunc} and Theorem \ref{thm:MB} on perturbed products of truncated unitary matrices and Muttalib-Borodin ensembles are similar to the Ginibre case, as we explain in Remark \ref{remark: analogy} without giving details.

\section{Proof of Theorem \ref{thm:zcmconvergence}}
\label{sec:zcmconvergence}

In this section we give the proof of Theorem \ref{thm:zcmconvergence}. We first need a general result about the zeros of the average characteristic polynomial.

\begin{lemma}
\label{lem:zerosavcharpol}
In a polynomial ensemble of the form \eqref{eq:pejpdf}, the average characteristic polynomial $p_n(z)=\mathbb E\left(\prod_{j=1}^n(z-x_j)\right)$ has $n$ simple real zeros.
\end{lemma}

\begin{proof}
The average 
characteristic polynomial has real coefficients and satisfies the orthogonality conditions \[\int_{\mathbb R}p_n(x)f_j(x)dx=0,\qquad j=0,\ldots, n-1.\]
If $p_n$ would have a non-simple or non-real zero $z_0$, we can write 
$p_n(z)=(z-z_0)(z-\bar z_0)\left(\sum_{k=0}^{n-2}a_k z^k\right)$ with $a_{n-2}=1$.
By the orthogonality conditions, we have \[\sum_{k=0}^{n-2}a_k\int_{\mathbb R} x^k\left(|x-z_0|^2f_j(x)\right)dx=0,\qquad j=0,\ldots, n-1.\] If this homogeneous linear system has a non-zero solution $(a_0,\ldots, a_{n-2})$, then the coefficient matrix is of rank $\leq n-2$, so the extended $n\times n$ matrix
$\left(\int|x-z_0|^2x^kf_j(x)dx\right)_{j,k=0,\ldots, n-1}$ is at most of rank $n-1$ and has zero determinant.

But on the other hand, by the Andreief identity, we have
\[\det\left(\int_{\mathbb R}|x-z_0|^2x^kf_j(x)dx\right)_{j,k=0,\ldots, n-1}=\frac{1}{n!}\int_{\mathbb R^n}\Delta(x)\det\left(f_{j-1}(x_k)\right)_{j,k=1,\ldots, n}\prod_{j=1}^{n}|x_j-z_0|^2dx_j,\]
which is strictly positive as it is equal to $\frac{Z_n}{n!}\mathbb E\left(\prod_{j=1}^{n}|x_j-z_0|^2\right)$. This is a contradiction, so $p_n$ has only simple real zeros.
\end{proof}

The proof of Theorem \ref{thm:zcmconvergence} now relies on the following two lemmas.

\begin{lemma}
\label{lem:convoffn}
Under the conditions of Theorem \ref{thm:zcmconvergence}, the sequence of functions \begin{equation}\label{eq:def fn}f_n(z):=\int\log(1-s/z)d(\mu_n-\mu)(s)\end{equation} 
converges to $0$, uniformly for $|z|>r$, where $r$ is such that $r > \max \{ |w| ~:~ w \in K \}$.
\end{lemma}

\begin{proof}
Point-wise convergence of $f_n$ to $0$ follows from the weak-$*$ convergence of $\mu_n$ to $\mu$, since $\log(1-s/z)$ is continuous for $|z|>r$. To prove uniform convergence, we note first that $f_n(z)$ is uniformly bounded: we have
\begin{equation}
|f_n(z)| \leq  \int |\log(1-s/z)| d\mu_n(s) +\int |\log(1-s/z)| d\mu(s)\leq 2 \max_{|z|\geq r, s\in K}|\log(1-s/z)|.
\end{equation}
If we define $h_n(z)=f_n(1/z)$, $h_n(z)$ is a uniformly bounded sequence of analytic functions on $|z|<1/r$, which converges point-wise to $0$. By Vitali's theorem, it follows that $h_n$ converges to $0$ uniformly for $|z|<1/r$, and hence $f_n(z)$ converges to $0$ uniformly for $|z|>r$.
\end{proof}

For any compactly supported probability measure $\mu$ on $\mathbb R$, we define
\begin{equation}\label{def g G}
g_\mu(z)=\int\log(z-s)d\mu(s),\qquad G_\mu(z)=g_\mu'(z)=\int\frac{1}{z-s}d\mu(s),
\end{equation} 
where we choose the logarithm corresponding to arguments between $-\pi$ and $\pi$.

\begin{lemma}
\label{lem:approxint}
Under the conditions of Theorem \ref{thm:zcmconvergence}, let $P_n(z)$ be the average characteristic polynomial of $S$.
For $|{\rm Re}\, z|$ large enough, we have 
\begin{equation}\label{lim log P}\lim_{n\to+\infty}\frac{1}{n}\log P_n(z) =g_\mu(s_c(z)) + \frac{1}{2\e^2} (z-s_c(z))^2,
\end{equation} where $s_c(z)$ is the solution in $s$ of $G_\mu(s) + \e^{-2} (s-z) = 0$.
\end{lemma}

\begin{proof}Writing $p_n$ for the average characteristic polynomial of $M$ and $\mu_n$ for its zero counting measure,
we have the identity
\begin{equation}
p_n(z)=\prod_{j=1}^n(z-z_j^{(n)})=e^{ng_{\mu_n}(z)},
\end{equation}
where $z_j^{(n)}$, $j=1,\ldots, n$ are the zeros of $p_n$. 

From the transformation formula (\ref{eq:sumacp}) for $P_n$ in terms of $p_n$, we have
\begin{equation}\label{int Pn}
P_n(z)=\frac{\sqrt{n}}{\sqrt{2\pi} i\e}\int_{i\mathbb R}e^{nF_n(s;z)}ds,\qquad F_n(s;z):=g_{\mu_n}(s)+\frac{1}{2\e^2} (z-s)^2.
\end{equation}
Since $g_{\mu_n}'(s)=\mathcal O(1/s)$ as $s\to\infty$, uniformly for $n$ large enough, the saddle point equation
$F_n'(s;z)=0$ has, for $z$ sufficiently large, a unique solution $s_c(z)$ such that $s_c(z)\sim z$ as $z\to\infty$, by the inverse function theorem.

In order to obtain large $n$ asymptotics for $P_n(z)$, 
we deform the integration contour $i\R$ in \eqref{int Pn} to the steepest descent path $\gamma=\gamma_n$ passing through $s_c(z)$ and on which the imaginary part of $F_n(s;z)$ is constant.
For $|{\rm Re}\, z|$ large, $\gamma$ makes a small angle with the vertical line trough $s_c(z)$, and therefore it remains outside of the compact $K$.
We have \[{\rm Re}\, F_n''(s;z)={\rm Re}\, g_{\mu_n}''(s) + \e^{-2} = \e^{-2} + \Or\left(s^{-2}\right),\qquad s \rightarrow \infty,\] and this implies that ${\rm Re}\, F_n$ achieves its unique local maximum on $\gamma$ at $s_c(z)$.
We can use the saddle point method to approximate the integral in \eqref{int Pn} in the following way, for $z$ sufficiently large, as $n\to\infty$.

We may choose an implicit parametrization $\gamma_n(t)$ of the steepest descent path $\gamma_n$ by imposing
\[F_n(\gamma_n(t);z)-F_n(s_c(z);z)=-t^2,\qquad t\in[-\delta,\delta], \qquad {\rm Im}\,\dot{\gamma_n}(0)>0,\]
for some sufficiently small $\delta>0$, and such that $|\dot{\gamma_n}(t)|=1$ for $|t|>\delta$. 
We can then write the integral in \eqref{int Pn} as
\begin{equation}\label{saddle point int1}
\int_{\gamma} e^{nF_n(s;z)} ds = e^{nF(s_c(z);z)}\int_{-\delta}^{\delta} \dot{\gamma_n}(t) e^{-nt^2} dt + \int_{\R\backslash [-\delta, \delta]} \dot{\gamma_n}(t) e^{nF_n(\gamma_n(t);z)} dt.
\end{equation}
Since ${\rm Re}\, F_n(s;z)$ has its unique global maximum on $\gamma_n$ at $t=0$ and grows as $t\to\pm\infty$, the second term is $\Or(e^{n[F_n(s_c(z))-\eta^2]})$ as $n\to+\infty$.
For the first term, note that $\dot{\gamma_n}(0)=i\sqrt{\frac{2}{F_n''(s_c(z);z)}}$. Moreover, there is a constant $C>0$ independent of $n$ such that
\[|\dot{\gamma_n}(t)-\dot{\gamma_n}(0)|\leq C|t|,\qquad t\in[-\delta,\delta].\]
From \eqref{saddle point int1}, we now get
\begin{multline}\label{saddle point int2}
\int_{\gamma} e^{nF_n(s;z)} ds = i\sqrt{\frac{2}{F_n''(s_c(z);z)}}e^{nF(s_c(z);z)}\int_{-\delta}^{\delta} e^{-nt^2} dt \\+\Or\left(\int_{-\delta}^{\delta} |t| e^{-nt^2} dt\right)+\Or\left(e^{n[F_n(s_c(z))-\eta^2]}\right),
\end{multline}
as $n\to\infty$.
Evaluating the integrals as $n\to\infty$ and substituting in \eqref{int Pn}, we finally obtain
\begin{equation}
P_n(z)\sim \frac{1}{\e} F_n''(s_c(z);z)^{-\frac{1}{2}} e^{nF_n(s_c(z);z)},\qquad n\to+\infty,
\end{equation}
for $|{\rm Re}\, z|$ sufficiently large.
It follows that 
\[\frac{1}{n}\log P_n(z) =g_{\mu_n}(s_c(z)) + \frac{1}{2\e^2} (z-s_c(z))^2+o(1),\qquad n\to+\infty.\]
Using the fact that 
\begin{equation}
g_{\mu_n(z)}=g_\mu(z)+f_n(z),\mbox{ with } f_n(z)=\int\log(1-s/z)d(\mu_n-\mu)(s)
\end{equation}
and Lemma \ref{lem:convoffn},
\eqref{lim log P} now follows easily.

\end{proof}

\begin{proof}[Proof of Theorem \ref{thm:zcmconvergence}]

Using the general definition of the free convolution, it was noted in \cite{Biane} that the free convolution of a compactly supported probability measure $\mu$ with the semi-circle $\lambda_\e$ satisfies the equation
\begin{equation}
G_{\mu\boxplus\lambda_\e}(s+\e^{2}G_\mu(s))=G_\mu(s),
\end{equation}
for $s$ sufficiently large.
In other words, if $s_c(z)$ is the solution of $G_\mu(s) + \e^{-2} (s-z) = 0$, we have
\begin{equation}\label{eq G free conv}
G_{\mu\boxplus\lambda_\e}(z)=G_\mu(s_c(z)).
\end{equation}
This implies, after a straightforward calculation, \begin{equation}\label{eq g free conv}
\frac{d}{dz}g_{\mu\boxplus\lambda_\e}(z)=\frac{d}{dz}\left(g_{\mu}(s_c(z))+\frac{1}{2\e^2} (z-s_c(z))^2\right),
\end{equation}
and upon integrating, we obtain
\begin{equation}\label{eq g free conv2}
g_{\mu\boxplus\lambda_\e}(z)=g_{\mu}(s_c(z))+\frac{1}{2\e^2} (z-s_c(z))^2,
\end{equation}
since both the left and the right hand side behave like $\log z + \Or(z^{-1})$ as $z\to\infty$.

By Lemma \ref{lem:approxint}, we have
\begin{equation}\label{logPn g}
\lim_{n\to +\infty}\frac{1}{n}\log P_n(z)=g_{\mu\boxplus\lambda_\e}(z)
\end{equation}
for $|{\rm Re}\, z|$ sufficiently large.
This implies that there exists an $n$-independent compact $\widetilde K$ such that $P_n$ has no zeros outside  $\widetilde K$, for $n$ sufficiently large. By Helly's theorem, $(\nu_n)_n$, and every subsequence of it, has a weak-$*$ converging subsequence. We claim that every such converging subsequence $(\nu_{n_k})_k$ converges to $\mu\boxplus\lambda_\e$. If so, it is easily seen by contraposition that the whole sequence $\nu_n$ converges to $\mu\boxplus\lambda_\e$.

To prove the claim, we suppose that a subsequence $\left(\nu_{n_k}\right)_k$ converges in weak-$*$ sense to some measure $\widetilde\nu$.
We then have
\begin{equation}
\lim_{k\to +\infty}\frac{1}{n_k}\log P_{n_k}(z)=\lim_{k\to+\infty}\int \log(z-s)d\nu_{n_k}(s)=g_{\widetilde \nu}(z),
\end{equation}
for $z$ outside $\widetilde K$,
and combining this with \eqref{logPn g},
it follows that $g_{\widetilde\nu}=g_{\mu\boxplus\lambda_\e}$.
Since the supports of $\widetilde\nu$ and $\mu\boxplus\lambda_\e$ are both contained in $\mathbb R$, 
$g_{\widetilde\nu}$ and $g_{\mu\boxplus\lambda_\e}$ are analytic in $\mathbb C\setminus\mathbb R$, and by analytic continuation we have in particular that ${\rm Re }\, g_{\widetilde\nu}={\rm Re }\, g_{\mu\boxplus\lambda_\e}$ everywhere in $\mathbb C$ except possibly on a set of $2$-dimensional Lebesgue measure $0$. We can then use the unicity theorem \cite[Theorem II.2.1]{SaffTotik} to conclude that $\widetilde\nu=\mu\boxplus\lambda_\e$.

\end{proof}

\section{Proof of Lemma \ref{thm:mainthm}}
\label{sec:proofmainthm}

We will now give the proof of the central lemma of this paper.
We assume that the conditions \eqref{eq:conv2} and \eqref{eq:growth2} hold.

Replacing $x$ and $y$ by $\frac{x}{cn^{\gamma}}$ and $\frac{y}{cn^{\gamma}}$ in \eqref{eq:kernelsum}, and substituting $s=\frac{u}{cn^{\gamma}}$ and $t=\frac{v}{cn^{\gamma}}$ in the integrals, we obtain the identity
\begin{equation}\label{eq:kernelsumzoomed}
\frac{1}{cn^\gamma}K_n^S\left(\frac{x}{cn^{\gamma}},\frac{y}{cn^{\gamma}}\right)=\frac{n^{1-2\gamma}}{2\pi i c^2\e_n^2} \int_{i\R} \int_{\R^+} \frac{1}{cn^{\gamma}} K_n\left(\frac{u}{cn^{\gamma}}, \frac{v}{cn^{\gamma}} \right) e^{\frac{n^{1-2\gamma}}{2c^2\e_n^2} \left( (x-u)^2 - (y-v)^2 \right)} dv du.
\end{equation}
By \eqref{eq:growth2}, we have 
\[\left| \frac{1}{cn^{\gamma}} K_n\left(\frac{u}{cn^{\gamma}}, \frac{v}{cn^{\gamma}} \right) \right| \leq \frac{c_1}{c^{1-\beta}} v^{-\beta} e^{\frac{c_2}{c}(|u|+|v|)},\] which implies that we can use the dominated convergence theorem in \eqref{eq:kernelsumzoomed} if $n\to+\infty$, $\e_n\to 0$ in such a way that $\lim_{n\to+\infty}c\e_n n^{\gamma-\frac{1}{2}}=\sigma>0$. Using \eqref{eq:conv2}, we immediately obtain the limit \eqref{eq:sumconv}, point-wise for $x, y\in\mathbb C$.

To see that the limit is uniform for $(x,y)$ in compact sets and to treat the case where $\lim_{n\to+\infty}\e_n n^{\gamma-\frac{1}{2}}=0$, we need the following technical estimates. 

\begin{lemma}
\label{lem:lemmamainthm}
Let $K_n : \C \times \R^+ \rightarrow \C$ be a sequence of kernels satisfying condition \eqref{eq:growth2}. Assume that $|x|, |y|\leq r$ and let $\omega > 0$. Then, for $n$ and $R$ sufficiently large, the following estimates hold:
\begin{multline}\label{estimate 1}
\left| \int_{i\R \backslash [-iR,iR]} \int_0^R \frac{1}{cn^{\gamma}} K_n\left(\frac{s}{cn^{\gamma}}, \frac{t}{cn^{\gamma}}\right) e^{\frac{1}{2\omega^2} ((x-s)^2 - (y-t)^2)} dtds \right| \\
\leq \frac{2 c_1}{c^{1-\beta}} \frac{R^{1-\beta}}{1-\beta} \frac{e^{2\frac{Rr}{\omega^2}+\frac{r^2}{\omega^2}+2\frac{c_2}{c} R}}{R-r-\frac{c_2}{c}\omega^2 } \omega^2 e^{-\frac{R^2}{2\omega^{2}}} ;
\end{multline}
\begin{equation}\label{estimate 2}
\left| \int_{-iR}^{iR} \int_R^{+\infty} \frac{1}{cn^{\gamma}} K_n\left(\frac{s}{cn^{\gamma}}, \frac{t}{cn^{\gamma}}\right) e^{\frac{1}{2\omega^{2}} ((x-s)^2 - (y-t)^2)} dtds \right| 
\leq \frac{c_1\sqrt{2\pi}}{c^{1-\beta}} \frac{e^{\frac{r^2}{\omega^2}+2\frac{Rr}{\omega^2}+2\frac{c_2}{c}R}}{R-r-\frac{c_2}{c}\omega^2}\omega^2 e^{-\frac{R^2}{2\omega^{2}}} ;
\end{equation}
\begin{equation}\label{estimate 3}
\left| \int_{i\R \backslash [-iR,iR]} \int_R^{+\infty} \frac{1}{cn^{\gamma}} K_n\left(\frac{s}{cn^{\gamma}}, \frac{t}{cn^{\gamma}}\right) e^{\frac{1}{2\omega^{2}} ((x-s)^2 - (y-t)^2)} dtds \right| 
\leq \frac{2c_1}{c^{1-\beta}} \frac{e^{2\frac{c_2}{c}R + 2\frac{Rr}{\omega^2} - \frac{R^2}{\omega^2}}}{\left( R - r - \frac{c_2}{c}\omega^2 \right)^2}\omega^2.
\end{equation}
\end{lemma}
\begin{proof}
To prove inequality \eqref{estimate 1}, we first use condition \eqref{eq:growth2} to get
\begin{multline}
\left| \int_{i\R \backslash [-iR,iR]} \int_0^R \frac{1}{cn^{\gamma}} K_n\left(\frac{s}{cn^{\gamma}}, \frac{t}{cn^{\gamma}}\right) e^{\frac{1}{2\omega^{2}} ((x-s)^2 - (y-t)^2)} dtds \right| \\
\leq \frac{c_1}{c^{1-\beta}} \int_{\R \backslash [-R,R]} e^{\frac{c_2}{c}  |s|} \left| e^{\frac{1}{2\omega^{2}} (x-is)^2} \right| ds \int_0^R t^{-\beta} e^{\frac{c_2}{c}  t}  e^{-\frac{1}{2\omega^{2}} (y-t)^2}  dt.
\end{multline}
The choices of $R$ and $r$ then allow us to estimate
\begin{equation}
\left| e^{\frac{1}{2\omega^{2}} (x-is)^2}\right| \leq e^{\frac{r^2}{2\omega^2}} e^{\frac{|s|r}{\omega^2}} e^{ - \frac{s^2}{2\omega^{2}}} \mbox{ and }  e^{-\frac{1}{2\omega^{2}} (y-t)^2} \leq e^{\frac{r^2}{2\omega^2}} e^{\frac{Rr}{\omega^2}} e^{-\frac{t^2}{2\omega^{2}}}.
\end{equation}
Using these two things, and the symmetry along the imaginary axis, we have
\begin{multline}
\left| \int_{i\R \backslash [-iR,iR]} \int_0^R \frac{1}{cn^{\gamma}} K_n\left(\frac{s}{cn^{\gamma}}, \frac{t}{cn^{\gamma}}\right) e^{\frac{1}{2\omega^{2}} ((x-s)^2 - (y-t)^2)} dtds \right| \\
\leq \frac{2 c_1}{c^{1-\beta}} \frac{R^{1-\beta}}{1-\beta} e^{\frac{r^2}{\omega^2}} e^{\frac{Rr}{\omega^2}} e^{\frac{c_2}{c} R} \int_R^{+\infty} e^{\frac{c_2}{c}  s} e^{\frac{1}{2\omega^{2}} (2sr-s^2)} ds.
\end{multline}
We may integrate by parts the integral at the right and doing so, we obtain
\begin{multline}
\label{eq:bigestimatelemmamainthm}
\int_R^{+\infty} e^{\frac{c_2}{c}  s} e^{\frac{1}{2\omega^{2}} (2sr-s^2)} ds
= \frac{\omega^2e^{\frac{c_2}{c}  R + \frac{Rr}{\omega^2} - \frac{R^2}{2\omega^{2}}}}{R - r - \frac{c_2}{c}\omega^2} - \int_R^{+\infty} \frac{\omega^4}{\left(s-r-\frac{c_2}{c}\omega^2\right)^2} e^{\frac{c_2}{c}  s + \frac{sr}{\omega^2} - \frac{s^2}{2\omega^{2}}} ds \\
\leq \frac{\omega^2 e^{\frac{c_2}{c}  R + \frac{Rr}{\omega^2} - \frac{R^2}{2\omega^{2}}}}{R - r - \frac{c_2}{c}\omega^2},
\end{multline}
which yields \eqref{estimate 1}.

\bigskip

The proofs for inequalities \eqref{estimate 2} and \eqref{estimate 3} are very similar. For \eqref{estimate 2}, the only thing that changes is that the roles of $s$ and $t$ are interchanged. For \eqref{estimate 3}, we need to use  inequality (\ref{eq:bigestimatelemmamainthm}) twice.
\end{proof}

We now proceed with the proof of Lemma \ref{thm:mainthm}.
To see that the limit \eqref{eq:sumconv} is uniform for $x,y$ in compact sets, we assume that $|x|,|y|<r$ and choose $R>r$, as in  Lemma \ref{lem:lemmamainthm}.
Define the three ensembles 
\begin{equation*}
U_1 := [-iR, iR] \times [R,+\infty), U_2 := \left(i\R \backslash [-iR, iR]\right) \times [0,R] \mbox{ and } U_3 := \left(i\R \backslash [-iR, iR]\right) \times [R, +\infty),
\end{equation*}
and write $\tilde{\e}_n := c\e_n n^{\gamma-\frac{1}{2}}$.
By \eqref{eq:kernelsumzoomed}, we have
\begin{multline}\label{last estimate lemma}
\left| \frac{1}{ cn^{\gamma}} K_n^S\left(\frac{x}{cn^{\gamma}}, \frac{y}{cn^{\gamma}}\right) - \frac{1}{2\pi i\sigma^2}\int_{i\R} \int_{\R^+} \K(s,t) e^{\frac{1}{2\sigma^2} ((x-s)^2 - (y-t)^2 )} dtds \right| \\
\leq \left| \frac{1}{2\pi i\tilde{\e}_n^2}\int_{-iR}^{iR} \int_0^R \left(v^{\beta}\frac{1}{cn^{\gamma}}K_n\left(\frac{s}{cn^{\gamma}}, \frac{t}{cn^{\gamma}}\right)- v^{\beta}\K(s,t)\right) v^{-\beta}e^{\frac{1}{2\tilde{\e}_n^2} ((x-s)^2 - (y-t)^2)} dtds  \right| \\
+\left| \int_{-iR}^{iR} \int_0^R \K(s,t)\left(\frac{1}{2\pi i\tilde{\e}_n^2} e^{\frac{1}{2\tilde{\e}_n^2} ((x-s)^2 - (y-t)^2)}-\frac{1}{2\pi i\sigma^2} e^{\frac{1}{2\sigma^2} ((x-s)^2 - (y-t)^2)} \right)dtds  \right| \\
 +\frac{1}{2\pi\tilde{\e}_n^2} \sum_{j = 1}^3  \left| \frac{1}{cn^{\gamma}}\iint_{U_j} K_n\left(\frac{s}{cn^{\gamma}}, \frac{t}{cn^{\gamma}}\right) e^{\frac{1}{2\tilde{\e}_n^2} ((x-s)^2 - (y-t)^2)} dsdt) \right| \\+ \frac{1}{2\pi\sigma^2}\sum_{j = 1}^3  \left| \iint_{U_j}  \K(s,t) e^{\frac{1}{2\sigma^2} ((x-s)^2 - (y-t)^2)} dsdt \right| .
\end{multline}
In the limit where $n\to+\infty$ and 
$\tilde{\e}_n=c\e_n n^{\gamma-1/2}\to \sigma>0$,
the first term at the right hand side of the above expression tends to $0$ uniformly for $|x|, |y|<r$ because of  \eqref{eq:conv2}. The second term at the right tends to $0$ as $n\to\infty$ by the dominated convergence theorem. The remaining terms at the right can be estimated using Lemma \ref{lem:lemmamainthm} with $\omega=\tilde{\e}_n$ and become small (uniformly in $x$ and $y$) as $R$ gets large. Since the left hand side does not depend on $R$, we can take $R$ as large as we want, and this implies that we have \eqref{eq:sumconv} uniformly for $|x|$ and $|y|$ smaller than $r$.

\bigskip

We now deal with the case where  $\tilde{\e}_n \to 0$ as $n\to+\infty$. By \eqref{eq:kernelsumzoomed}, we have
\begin{multline}
\frac{1}{cn^\gamma}K_n^S(\frac{x}{cn^{\gamma}},\frac{y}{cn^{\gamma}})  \\=
\frac{1}{2\pi i \tilde{\e}_n^2} \int_{[-iR,iR]} \int_{[0,R]} \frac{1}{cn^{\gamma}} K_n\left( \frac{u}{cn^{\gamma}}, \frac{v}{cn^{\gamma}} \right) e^{\frac{1}{2\tilde{\e}_n^2} \left( (x-u)^2 - (y-v)^2 \right)} dvdu \\
+ \sum_{j=1}^3 \frac{1}{2\pi i \tilde{\e}_n^2} \iint_{U_j} \frac{1}{cn^{\gamma}} K_n\left( \frac{u}{cn^{\gamma}}, \frac{v}{cn^{\gamma}} \right) e^{\frac{1}{2\tilde{\e}_n^2} \left( (x-u)^2 - (y-v)^2 \right)} dvdu,
\end{multline}
where we have cut the integral in the same $4$ parts as before.
We may use Lemma \ref{lem:lemmamainthm} with $\omega = \tilde{\e}_n$ and we get that the three terms in the sum on the last line tend to $0$ as $n\to+\infty$, provided $R$ is sufficiently large. In the first term, we can deform the integration contour for $u$ to $\sqsupset:=[-iR, x-iR] \cup [x-iR,x+iR] \cup [x+iR,iR]$. This does not change the integral since $K_n$ is analytic in $u$. We use the convergence in \eqref{eq:conv2}, which is uniform for $u$ and $v$ on their respective integration contours, and obtain
\begin{multline}
\frac{1}{2\pi i \tilde{\e}_n^2} \int_{[-iR,iR]} \int_{[0,R]} \frac{1}{cn^{\gamma}} K_n\left( \frac{u}{cn^{\gamma}}, \frac{v}{cn^{\gamma}} \right) e^{\frac{1}{2\tilde{\e}_n^2} \left( (x-u)^2 - (y-v)^2 \right)} dvdu = \\
\frac{1}{2\pi i \tilde{\e}_n^2} \int_{\sqsupset} \int_{[0,R]} \left(\K(u,v) + o(1)v^{-\beta}\right) e^{\frac{1}{2\tilde{\e}_n^2} \left( (x-u)^2 - (y-v)^2 \right)} dvdu , \mbox{ as } n \rightarrow + \infty.
\end{multline}
Since $R$ is arbitrary, we can take it large enough so that the contribution to the integral on the paths $[-iR, x-iR]$ and $[x+iR,iR]$ is small in $n$. Then, the contours $[x-iR,x+iR]$ and $[0,R]$ are the steepest descent contours of the $u$- and $v$-phase functions, respectively. This allows us to apply the usual saddle point method to the integral, which gives
\begin{equation}
\frac{1}{2\pi i \tilde{\e}_n^2} \int_{[x-iR,x+iR]} \int_{[0,R]} \left(\K(u,v)+o(1)v^{-\beta}\right) e^{\frac{1}{2\tilde{\e}_n^2} \left( (x-u)^2 - (y-v)^2 \right)} dvdu = \K(x,y)+o(1), \mbox{ as } n \rightarrow +\infty,
\end{equation}
for $x,y>0$ fixed. Note that the error term is not uniform as $y\to 0$.
This completes the proof of Lemma \ref{thm:mainthm}.

\section{Proofs of applications}

\subsection{Proof of Theorem \ref{cor:luegeneralresult}: sub-critical and critical cases}
\label{sec:lueproofs}

For the proof of parts (i) and (ii) of Theorem \ref{cor:luegeneralresult}, we only need to verify that the conditions of Lemma \ref{thm:mainthm}, namely \eqref{eq:conv2}-\eqref{eq:growth2},  are satisfied. It should be noted that it is important to have uniformity of  \eqref{eq:conv2} for $u$ in any compact subset of $\mathbb C$ and $v$ in any compact subset of $[0,+\infty)$, and therefore \eqref{eq:conv2} is not a direct consequence of \eqref{eq:luehardedgescaling}.

\bigskip

Define $p_n$, $n = 0,1,...$ to be the normalized orthogonal polynomials with respect to the generalized Laguerre weight $w(x) = x^{\alpha} e^{-nx^k}$. Let $Y_n$ be the matrix-valued function
\begin{equation}\label{def Y}
Y_n(z) = \begin{pmatrix}
\frac{1}{\kappa_n} p_n(z) & \frac{1}{\kappa_n} C(p_n w)(z) \\
-2\pi i \kappa_{n-1} p_{n-1}(z) & -2\pi i \kappa_{n-1} C(p_{n-1} w)(z)
\end{pmatrix},
\end{equation}
where $\kappa_j>0$ is the leading coefficient of $p_j$ and $Cf$ is the Cauchy transform 
\begin{equation}
\label{eq:cauchytransform}
Cf(z) = \frac{1}{2\pi i}\int_{0}^{+\infty} \frac{ f(s) }{s-z} ds, \mbox{ for } z \notin [0,+\infty).
\end{equation}
The matrix $Y_n$ is the solution to the usual Fokas-Its-Kitaev Riemann-Hilbert problem for orthogonal polynomials \cite{FokasItsKitaev}. The eigenvalue correlation kernel $K_n$ for a (generalized) LUE matrix is expressed in terms of $Y_n$ as
\begin{equation}
\label{eq:luekernel}
{K}_n(x,y) = w(y) \frac{1}{2\pi i (x-y)} \begin{pmatrix}0&1\end{pmatrix} Y_n^{-1}(y)Y_n(x) \begin{pmatrix} 1 \\ 0 \end{pmatrix},
\end{equation}
which is easily verified by the Christoffel-Darboux formula and the fact that $\det Y_n = 1$.

\bigskip

Define the re-scaled matrix
\begin{equation}
\label{eq:RHUsolution}
U_n(z) = \left( \frac{2}{k A_k} \right)^{\frac{1}{k}\left(-n-\frac{\alpha}{2}\right) \sigma_3} Y_n\left( \left( \frac{2}{k A_k} \right)^{\frac{1}{k}} z \right) \left( \frac{2}{k A_k} \right)^{\frac{1}{k}\frac{\alpha}{2} \sigma_3},\qquad A_k = \prod_{j=1}^k \frac{2j-1}{2j},
\end{equation}
where $\sigma_3 = \begin{pmatrix} 1 & 0 \\ 0 & -1 \end{pmatrix}$. 
The rescaled kernel $\widetilde{K}_n$ defined by
\begin{equation}\label{def rescaled kernel U}
\widetilde{K}_n(x,y) = \left( \frac{2}{k A_k} \right)^{\frac{\alpha}{k}}y^{\alpha} e^{- \frac{2n}{k A_k} y^k} \frac{1}{2\pi i (x-y)} \begin{pmatrix}0&1\end{pmatrix} U_n^{-1}(y)U_n(x) \begin{pmatrix} 1 \\ 0 \end{pmatrix},
\end{equation}
corresponds to a re-scaled LUE in which the limiting mean eigenvalue distribution is supported on $[0,1]$, and given by \cite{Vanlessen}
\begin{equation}\label{LUE lim density}
d\widetilde\mu(x)=\frac{1}{2\pi}\sqrt{\frac{1-x}{x}}\widetilde h(x)dx,\qquad \widetilde h(x) = 2 \sum_{j=0}^{k-1} \frac{A_{k-1-j}}{A_k} x^j,\qquad x\in (0,1)
.
\end{equation}
We have \begin{equation}\label{rescaled kernel}\widetilde{K}_n(x,y) = \left( \frac{2}{k A_k} \right)^{\frac{\alpha+1}{k}} K_n\left( \left( \frac{2}{k A_k} \right)^{\frac{1}{k}} x, \left( \frac{2}{k A_k} \right)^{\frac{1}{k}} y \right).\end{equation}

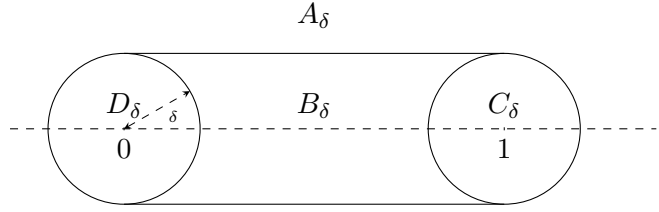
\begin{figure}
\centering
\begin{tikzpicture}[scale = 1, > = stealth]

\draw[thin, dashed] (-1.5, 0) -- (7,0);

\draw (0,0) circle (1cm);
\draw[very thin] (0,-.02) -- node[anchor = north] {\small{0}} (0,.02);
\draw (0,0) node[anchor = south] {\small{$D_{\delta}$}};

\draw (5,0) circle (1cm);
\draw[very thin] (5,-.02) -- node[anchor = north] {\small{1}} (5,.02);
\draw (5,0) node[anchor = south] {\small{$C_{\delta}$}};

\draw (90:1cm) -- (5,1);
\draw (270:1cm) -- (5,-1);
\draw (2.5,0) node[anchor = south] {\small{$B_{\delta}$}};

\draw (2.5,1.5) node {\small{$A_{\delta}$}};

\draw[<->, very thin, dashed] (0,0) -- node[anchor = north, near end] {\tiny{$\delta$}} (30:1cm);
\end{tikzpicture}
\caption{The four regions in which the asymptotics of $U_n$ are expressed differently.}
\label{fig:fourregions}
\end{figure}

We now describe the asymptotics for $U_n(z)$, which were obtained in \cite{Vanlessen}.
Let $\delta> 0$ be sufficiently small and consider the following four regions, as illustrated in Figure \ref{fig:fourregions}:
\begin{align}&D_{\delta} = \{ z \in \C ~:~ |z| < \delta \},&&C_{\delta} = \{ z \in \C ~:~ |z-1| < \delta \},\\
&B_{\delta} = \{ z \in \C ~:~ 0 < {\rm Re}\, z < 1, |\Im(z)| < \delta, z \notin D_{\delta} \cup C_{\delta} \},&& A_{\delta} = \C \backslash (B_{\delta} \cup C_{\delta} \cup D_{\delta}).
\end{align}
The nature of the asymptotics for $U_n(z)$ is different in each of those regions.

For $z \in A_{\delta}$, we have
\begin{equation}\label{asymptotics U A}
U_n(z)\begin{pmatrix} 1 \\ 0 \end{pmatrix} = \frac{z^{-\frac{\alpha}{2}}}{2z^{\frac{1}{4}}(z-1)^{\frac{1}{4}}}e^{ng_{\widetilde\mu}(z)} \begin{pmatrix}  2^{-\alpha} \varphi(z)^{\frac{1}{2}(\alpha+1)} (1 + \Or(1/n)) \\ -i e^{- n \ell} 2^{\alpha} \varphi(z)^{\frac{1}{2} (\alpha-1)} (1+\Or(1/n)) \end{pmatrix},
\end{equation}
as $n\to\infty$, where  $\ell = -\frac{2}{k} - 4 \log 2$, 
\begin{equation}
g_{\widetilde\mu}(z)=\int\log(z-s)d\widetilde\mu(s)\sim \log z,\qquad\mbox{ as $z\to\infty$},
\end{equation}
and
\begin{equation}
\varphi(z) = 2(z-1/2)+2z^{1/2}(z-1)^{1/2},
\end{equation}
with the square roots defined such that $\varphi$ is analytic except on $[0,1]$.
%
It follows that there exists a constant $c_1$ such that
\begin{equation}\label{estimate 1 U}
\left\|U_n(z)\begin{pmatrix} 1 \\ 0 \end{pmatrix}\right\|\leq e^{c_1 n |z|},\qquad z\in A_\delta,
\end{equation}
for $n$ sufficiently large. This is a very rough estimate but it will be enough for our purposes.

For $z \in B_{\delta}$, we have
\begin{equation}
U_n(z)\begin{pmatrix} 1 \\ 0 \end{pmatrix} = \frac{z^{-\frac{\alpha}{2}}}{z^{\frac{1}{4}}(z-1)^{\frac{1}{4}}} e^{n\frac{z^k}{kA_k}} \begin{pmatrix} e^{\frac{1}{2}n\ell} 2^{-\alpha} \cos \left( \eta_1(z) - i n \xi(z) - \frac{\pi}{4} \right) (1+\Or(1/n)) \\ e^{-\frac{1}{2}n\ell}2^{\alpha} \cos \left( \eta_2(z) - i n \xi(z) - \frac{\pi}{4} \right) (1+\Or(1/n)) \end{pmatrix},
\end{equation}
as $n\to\infty$,
for some continuous functions $\eta_1, \eta_2$ independent of $n$, and with
\begin{equation}\xi(z) = -\pi i \int_1^{z} \frac{1}{2\pi i} \frac{(s-1)^{1/2}}{s^{1/2}} \widetilde h(s) ds,
\end{equation}
with the square roots such that $\xi$ is analytic in $\mathbb C\setminus[0,1]$.
It follows that there exists a constant $c_2$ such that
\begin{equation}\label{estimate 2 U}
\left\|U_n(z)\begin{pmatrix} 1 \\ 0 \end{pmatrix}\right\|\leq e^{c_2 n|z|},\qquad z\in B_\delta,
\end{equation}
for $n$ sufficiently large.

For $z \in C_{\delta}$,
\begin{multline*}
U_n(z)\begin{pmatrix} 1 \\ 0 \end{pmatrix} = \frac{\sqrt{\pi} z^{-\frac{\alpha}{2}}}{z^{1/4}(z-1)^{1/4}} e^{ n \frac{z^k}{kA_k}} \\
\times \begin{pmatrix} e^{\frac{1}{2}n\ell}2^{-\alpha} \left( \cos(\eta_1(z)) f_n(z)^{1/4} \Ai(f_n(z)) - i \sin(\eta_1(z)) f_n(z)^{-1/4} \Ai'(f_n(z)) \right) (1+\Or(1/n)) \\ e^{-\frac{1}{2}n\ell} 2^{\alpha} \left( -i \cos(\eta_2(z)) f_n(z)^{1/4} \Ai(f_n(z)) - \sin(\eta_2(z)) f_n(z)^{-1/4} \Ai'(f_n(z)) \right) (1+\Or(1/n)) \end{pmatrix},
\end{multline*}
as $n\to\infty$, with $\Ai$ the Airy function, and
 $f_n$ a conformal map defined in a neighborhood of $1$, satisfying $f_n(1)=0$, $f_n'(1)>0$ and
\begin{equation}
\label{eq:fn}
\frac{2}{3} f_n(z)^{\frac{3}{2}} = n\xi(z).
\end{equation}
Using the asymptotics for the Airy function at large arguments and the fact that $|f_n(z)|\leq Cn^{2/3}|z-1|$ for some constant $C>0$ independent of $n$, for $z$ sufficiently close to $1$,
it follows again that there exists a constant $c_3$ such that
\begin{equation}\label{estimate 3 U}
\left\|U_n(z)\begin{pmatrix} 1 \\ 0 \end{pmatrix}\right\|\leq e^{c_3 n|z|},\qquad z\in C_\delta,
\end{equation}
for $n$ sufficiently large.

Finally, we have for $z \in D_{\delta}$,
\begin{multline}\label{as U 0}
U_n(z)\begin{pmatrix} 1 \\ 0 \end{pmatrix} = \frac{(-1)^n \sqrt{\pi} (-\widetilde{f}_n(z))^{1/4} z^{-\frac{\alpha}{2}}}{z^{1/4}(1-z)^{1/4}} e^{n\frac{z^k}{kA_k}} \\
\times \begin{pmatrix} e^{\frac{1}{2}n\ell} 2^{-\alpha} \left( \sin(\zeta_1(z)) J_{\alpha} (2(-\widetilde{f}_n(z))^{1/2}) + \cos(\zeta_1(z)) J_{\alpha}'(2(-\widetilde{f}_n(z))^{1/2}) \right) (1+\Or(1/n)) \\ -i e^{-\frac{1}{2}n\ell} 2^{\alpha} \left( \sin(\zeta_2(z)) J_{\alpha} (2(-\widetilde{f}_n(z))^{1/2}) + \cos(\zeta_2(z)) J_{\alpha}'(2(-\widetilde{f}_n(z))^{1/2}) \right) (1+\Or(1/n)) \end{pmatrix},
\end{multline}
as $n\to\infty$, with $\widetilde f_n$ a conformal map defined in a neighbourhood of $0$ and satisfying
\begin{equation}
\label{eq:fntilde}
2\widetilde{f}_n(z)^{\frac{1}{2}} = n\xi(z),\qquad \widetilde f_n(0)=0,\qquad \widetilde f_n'(0)<0.
\end{equation}
The functions $\zeta_1$ and $\zeta_2$ are independent of $n$ and such that $\zeta_1(0)=-\zeta_2(0)=\pm \frac{\pi}{2}$.
We have $|\widetilde{f}_n(z)|\leq Cn^2|z|$ for some positive $n$-independent constant $C$, if $z$ is sufficiently close to $0$. The asymptotics \cite[Formula 10.7.8]{NIST} of the Bessel function $J_{\alpha}$ and the fact that $z^{-\alpha}J_\alpha(z)$ is bounded allow us to conclude that there exist constants $c_4, C>0$ such that 
\begin{equation}\label{estimate 4 U}
\left\|e^{-\frac{n}{2}\ell\sigma_3}U_n(z)\begin{pmatrix} 1 \\ 0 \end{pmatrix}\right\| \leq Ce^{c_4 n|z|},\qquad z\in D_\delta.
\end{equation}
for $n$ sufficiently large.

\begin{proof}[Proof of Theorem \ref{cor:luegeneralresult}, parts (i) and (ii)]

The hard edge scaling limit \eqref{eq:luehardedgescaling} was proved in \cite{Vanlessen} (using \eqref{as U 0}) to hold point-wise for $u, v>0$. From the proof, or from \eqref{as U 0}, it is however readily seen that it holds point-wise for any $u\in\mathbb C$. For $\alpha\geq 0$, the Bessel function $J_\alpha$ is bounded near $0$, and then it is straightforward to show, again using \eqref{as U 0}, that \eqref{eq:conv2} with $\beta=0$ is uniform for $u$ in any compact in $\mathbb C$ and for $v$ in any compact in $[0,+\infty)$.
If $-1<\alpha<0$, $J_\alpha(z)$ blows up as $z\to 0$, but we have that $z^{-\alpha}J_\alpha(z)$ is analytic at $0$, and this can be used to show that  \eqref{eq:conv2} with $\beta=-\alpha$ holds uniformly for $u$ in any compact in $\mathbb C$ and for $v$ in any compact in $[0,+\infty)$.

To verify condition \eqref{eq:growth2} for $(u,v) \in i\R \times \R^+$, we have to consider eight different regions for $(u,v)$: $A_\delta\times A_\delta$, $A_\delta\times B_\delta$, $A_\delta\times C_\delta$, $A_\delta\times D_\delta$, $D_\delta\times A_\delta$, $D_\delta\times B_\delta$, $D_\delta\times C_\delta$, and $D_\delta\times D_\delta$.
In each of these cases, we can bound the kernel $\widetilde K_n$ defined in \eqref{def rescaled kernel U} using the estimates \eqref{estimate 1 U}, \eqref{estimate 2 U}, \eqref{estimate 3 U}, and \eqref{estimate 4 U} corresponding to the different regions. Substituting them in \eqref{rescaled kernel}, we obtain the desired estimate \eqref{eq:growth2} in each of the regions. 

The results now follow directly from Lemma \ref{thm:mainthm}.
\end{proof}

\subsection{Proof of Theorem \ref{cor:luegeneralresult}: super-critical case}
\label{sec:lueproofs2}

For part (iii) of Theorem \ref{cor:luegeneralresult}, we will use saddle point methods and a modified version of the integral representation \eqref{eq:kernelsum} for the correlation kernel $K_n^S$.
We first need some technical results.

\begin{lemma}
\label{lem:edgeofthesupport}
Let $\mu$ be the limiting mean eigenvalue distribution given by \eqref{eq:luemacro}, and let $\lambda_{\e}$ be the rescaled semicircle law \eqref{eq:rescaledsemicircle}.
The free additive convolution $\mu\boxplus\lambda_\e$ is supported on an interval $[a_\e, b_\e]$, and $a_\e$ is given by
\begin{equation}
\label{eq:criticalpointlimitingdensity}
a_{\e} = u_{\e} - \e^2 \int \frac{d\mu(x)}{x-u_{\e}} ,
\end{equation}
where $u_{\e}$ is the unique negative solution of the equation
\begin{equation}
\label{eq:conditiononedges}
\int \frac{d\mu(x)}{(x-u_{\e})^2}  = \frac{1}{\e^2}.
\end{equation}
Moreover, for some $\kappa,\widehat\kappa>0$,
\begin{equation}\label{a u epsilon}
u_\e= -\kappa \e^{4/3} (1+o(1)),\qquad a_\e=-\widehat \kappa\e^{4/3}(1+o(1)), \qquad \e\to 0.
\end{equation}
\end{lemma}

\begin{proof}
It follows from the general results in \cite{Biane} on the free additive convolution of a measure with a semi-circle law that the left endpoint $a_\e$ is given by 
 \eqref{eq:criticalpointlimitingdensity}, with $u_\e$ solving \eqref{eq:conditiononedges}, provided that such a solution $u_\e$ exists.
For general $\mu$ and $\epsilon>0$,
it is not always true that \eqref{eq:conditiononedges} has a unique solution. However,
 with $\mu$ given by \eqref{eq:luemacro}, the left hand side of \eqref{eq:conditiononedges} tends to $0$ as $u_\e\to -\infty$, is increasing, and tends to $+\infty$ as $u_\e\to 0$. This implies that \eqref{eq:conditiononedges} has indeed a unique solution for any $\e>0$.
Small $\e$ asymptotics for $u_\e$ and $a_\e$ can be obtained from \eqref{eq:conditiononedges} and \eqref{eq:criticalpointlimitingdensity} using residue arguments. This leads to \eqref{a u epsilon}.
\end{proof}

The following result provides us with suitable paths in the complex plane where the real part of a certain phase function $\Psi_\e$, to be used in the saddle point analysis later on, is monotone.

\begin{lemma}
\label{lem:steepestdescentcontours}
Let $\mu$ be given by \eqref{eq:luemacro}, let $\e>0$, and let $a_{\e}$ and $u_\e$ be as in Lemma \ref{lem:edgeofthesupport}. Define 
\begin{equation}
\label{eq:phasefct}
\Psi_{\e}(z) = \frac{1}{2\e^2} (z-a_{\e})^2 + \int \log(z-x) d\mu(x),\qquad z\in\mathbb C\setminus(-\infty,b],
\end{equation} and define two paths $\gamma_1$ and $\gamma_2$ by
\begin{equation}
\gamma_1(t) = u_\e+ e^{\frac{\pi i}{3}}|u_\e|+ t,\qquad 
\gamma_2(t) = u_\e+e^{\frac{2\pi i}{3}}|u_\e| +it,\qquad t>0.
\end{equation}
Then, for $\e>0$ sufficiently small, the functions
\begin{equation*}
t\in\R^+ \mapsto {\rm Re}\,( \Psi_{\e}(\gamma_1(t)) ) \mbox{ and } t \in \R^+ \mapsto {\rm Re}\,( \Psi_{\e}(\bar{\gamma}_1(t)) )
\end{equation*}
are increasing and the functions
\begin{equation*}
t \in \R^+ \mapsto {\rm Re}\,( \Psi_{\e}(\gamma_2(t)) ) \mbox{ and } t \in \R^+ \mapsto {\rm Re}\,( \Psi_{\e}(\bar{\gamma}_2(t)) )
\end{equation*}
are decreasing.
\end{lemma}

\begin{proof}

Because of conjugational symmetry, it suffices to check that
\begin{equation*}
\frac{d}{dt} {\rm Re}\,(\Psi_{\e}(\gamma_1(t))) > 0 \mbox{ and } \frac{d}{dt} {\rm Re}\,(\Psi_{\e}(\gamma_2(t))) < 0,\qquad \mbox{for $t>0$.}
\end{equation*}
On $\gamma_2$, we have after a straightforward computation,
\begin{equation}\label{Psi1}
\frac{d}{dt} {\rm Re}\,(\Psi_{\e}(\gamma_2(t))) = \left( t+\frac{\sqrt{3}}{2} |u_{\e}| \right) \left( \frac{-1}{\e^2} + \int \frac{ d\mu(s) }{\left(\frac{3}{2}u_{\e} - s\right)^2 + \left(t+\frac{\sqrt{3}}{2} |u_{\e}| \right)^2} \right).
\end{equation}
For $t\geq 0$, we have
\begin{equation}
0\leq \int \frac{ d\mu(s) }{\left(\frac{3}{2}u_{\e} - s\right)^2 + \left(t+\frac{\sqrt{3}}{2} |u_{\e}| \right)^2} \leq \int \frac{ d\mu(s) }{\left(\frac{3}{2}u_{\e} - s\right)^2 + \frac{3}{4} u_{\e}^2}.
\end{equation}
The integral at the right is easily computable using the residue theorem, and by \eqref{a u epsilon}, it is readily seen that it is of order $\e^{-4/3}$ as $\e \rightarrow 0$. Hence, by \eqref{Psi1}, for $\e$ small enough,
\begin{equation}
\frac{d}{dt} {\rm Re}\,(\Psi_{\e}(\gamma_2(t))) < 0,\qquad \mbox{for all $t\geq 0$.}
\end{equation}

\bigskip

On $\gamma_1$, we have
\begin{equation}
\frac{d}{dt} {\rm Re}\,\left( \Psi_{\e}(\gamma_1(t)) \right) = \frac{1}{\e^2} \left(\frac{u_{\e}}{2}+t-a_{\e}\right) + {\rm Re}\,\left( \int \frac{d\mu(x)}{u_\e+ e^{\frac{\pi i}{3}}|u_\e|+ t - x} \right).
\end{equation}
By Lemma \ref{lem:edgeofthesupport}, we can replace $a_{\e}$ by its expression in terms of $u_{\e}$ given in  (\ref{eq:criticalpointlimitingdensity}), and we get
\begin{equation}\label{derPsi1}
\frac{d}{dt}  {\rm Re}\,\left( \Psi_{\e}(\gamma_1(t)) \right)  = \frac{1}{\e^2} \left(-\frac{u_{\e}}{2}+t\right) + \int \frac{d\mu(x)}{x-u_{\e}}  + {\rm Re}\,\left( \int \frac{d\mu(x)}{u_\e+ e^{\frac{\pi i}{3}}|u_\e|+ t - x} \right).
\end{equation}
This is easily seen to be positive if $t \geq b-\frac{u_{\e}}{2}$. For $t$ smaller, we need to be more careful and exploit the fact that $\e$ is small.
The Stieltjes transform can be computed explicitly by \eqref{LUE lim density} and a residue calculus (or alternatively by observing that $\mu$ is an equilibrium measure which satisfies Euler-Lagrange variational conditions): for $z \in \C \backslash [0,b]$,
\begin{equation}\label{Stieltjes}
\int \frac{d\mu(x)}{z-x} = \frac{k}{2}z^{k-1} - \frac{h(z)}{2} \frac{\sqrt{z-b}}{\sqrt{z}},
\end{equation}
with principal branches of the square roots.
Evaluating at $z=u_\e$ and at $z=t+u_\e+e^{\frac{\pi i}{3}}|u_\e|$, and substituting in \eqref{derPsi1}, we obtain 
\begin{multline}\label{derPsi2}
\frac{d}{dt}  {\rm Re}\,\left( \Psi_{\e}(\gamma_1(t)) \right)  = \frac{1}{\e^2} \left(-\frac{u_{\e}}{2}+t\right) + \frac{h(0)}{2} \sqrt{\frac{b}{|u_{\e}|}}  \\- {\rm Re}\,\left(\frac{h(t)}{2} \frac{\sqrt{t+u_\e+e^{\frac{\pi i}{3}}|u_\e|-b}}{\sqrt{t+u_\e+e^{\frac{\pi i}{3}}|u_\e|}}  \right) + \mathcal{O}(1), \qquad\textrm{ as } \e \rightarrow 0.
\end{multline}
For $t\gg\e^{4/3}$, the leading order term in this expansion is $\frac{t}{\e^2}$, and it then follows that $\frac{d}{dt}  {\rm Re}\,\left( \Psi_{\e}(\gamma_1(t)) \right)>0$ for $\e$ sufficiently small.

Writing $t=T|u_\e|\sim \kappa T\e^{4/3}$ as $\e\to 0$, we have
\begin{multline}\label{derPsi3}
\frac{d}{dt}  {\rm Re}\,\left( \Psi_{\e}(\gamma_1(t)) \right)  =\e^{-2/3}\left(\frac{\kappa}{2}+T\kappa+\frac{h(0)\sqrt{b}}{2\sqrt{\kappa}} \left( 1 - {\rm Re}\,\frac{i\sqrt{T-1+e^{-\frac{\pi i}{3}}}}{|T-1+e^{\frac{\pi i}{3}}|} \right)\right) \\+ o(\e^{-2/3}), \qquad\textrm{ as } \e \rightarrow 0.
\end{multline}
It is straightforward to check by a trigonometric argument that ${\rm Re}\,\frac{i\sqrt{T-1+e^{-\frac{\pi i}{3}}}}{|T-1+e^{\frac{\pi i}{3}}|}<1$ for $T>0$, and this implies by \eqref{derPsi3} that $\frac{d}{dt}  {\rm Re}\,\left( \Psi_{\e}(\gamma_1(t)) \right)>0 $ for $t>0$.
\end{proof}

We can now prove part (iii) of Theorem \ref{cor:luegeneralresult}. The general strategy of the proof is similar to the one of \cite[Theorem 1.2]{LiuWangZhang}.

\begin{proof}[Proof of Theorem \ref{cor:luegeneralresult}, part (iii)]
We use an alternative expression for $K_n^S$ given by \cite[Formula (3.23)]{ClaeysKuijlaarsWang}
\begin{equation}
\label{eq:sumkernelalternative}
K_n^S(x,y) = \frac{n}{(2\pi i)^2 \e^2} \int_C ds \int_{\Sigma} d\zeta \frac{1}{s-\zeta} \begin{pmatrix} 0 & 1 \end{pmatrix} Y_n^{-1}(s) Y_n(\zeta) \begin{pmatrix} 0 \\ 1 \end{pmatrix} e^{\frac{n}{2\e^2} \left( (x-s)^2 - (y-\zeta)^2 \right) },
\end{equation}
where $Y_n$ is defined in \eqref{def Y}.
Here, $\Sigma$ is a contour leaving from $+\infty$ in the upper half plane, encircling the positive real axis and going back to $+\infty$ in the lower half plane, as in Section \ref{sec:prodginibre}. The second contour $C$ does not intersect with $\Sigma$ but can otherwise be any path going from $c_1-i\infty$ to $c_2+i\infty$. The fact that the two contours do not intersect, explains why the first term of \cite[Formula (3.23)]{ClaeysKuijlaarsWang} cancels out here. We will choose the contours $C$ and $\Sigma$ such that they are suitable for a saddle point analysis of the integral in \eqref{eq:sumkernelalternative}.

\bigskip

We let $\Psi_\e$, which will serve as a phase function in the saddle point analysis, be defined as in (\ref{eq:phasefct}), and we let $u_{\e} < 0$ be as in Lemma \ref{lem:edgeofthesupport}. By (\ref{eq:conditiononedges}) and (\ref{eq:criticalpointlimitingdensity}), we have that $\Psi_{\e}'(u_{\e}) = \Psi_{\e}''(u_{\e})=0$ and since $u_{\e} < 0$, $\Psi_{\e}'''(u_{\e})=G_\mu''(u_\e) < 0$. Thus the steepest descent paths for $\Psi_\e$ emanating from $u_{\e}$ make angles respectively $0, \frac{2\pi}{3}$ and $-\frac{2\pi}{3}$. By \eqref{a u epsilon} and \eqref{Stieltjes}, we have
\begin{equation}
\label{eq:phfctderivatives}
 \Psi_{\e}^{(j)}(u_{\e})= -\kappa_j \e^{-\frac{2}{3}(2j-1)} (1+o(1)), \mbox{ as } \e \rightarrow 0,
\end{equation}
for $j\geq 3$, where $\kappa_3, \kappa_4, \ldots$ are positive constants independent of $\e$.

\bigskip

\begin{figure}
\centering
\begin{tikzpicture}[>= stealth, scale = .7]

\draw[very thin, ->] (-3, 0) -- (3,0);
\draw (3,0) node [anchor = south] {$\R$};
\draw[very thin, ->] (2, -2) -- (2,2);
\draw (2,2) node [anchor = west] {$i\R$};

\draw (60:1cm) -- (60:2cm);
\draw[->- = .7] (300:1cm) -- (60:1cm);
\draw (300:2cm) -- (300:1cm);
\draw (300:2cm) node [anchor = north west] {$\Sigma_{\textrm{in},n}$};

\draw (120:1cm) -- (120:2cm);
\draw[->- = .7] (240:1cm) -- (120:1cm);
\draw (240:2cm) -- (240:1cm);
\draw (120:2cm) node [anchor = south east] {$C_{\textrm{in},n}$};

\draw[very thin] (0,-.05) -- (0,.05);
\draw (0,0) node [anchor = north] {$u_{\e}$};


\draw[dashed, thin] (0,0) circle (2cm);
\draw (-2,0) node[anchor = north] {$2u_{\e}$};

\draw[dashed, thin] (1.5,0) arc (0:60:1.5cm);
\draw (30:1.5cm)  node[anchor = west] {$\frac{\pi}{3}$};
\end{tikzpicture}
\caption{Local contours around the saddle point of the phase function $\Psi_{\e}$.}
\label{fig:airysaddlepointcontours}
\end{figure}
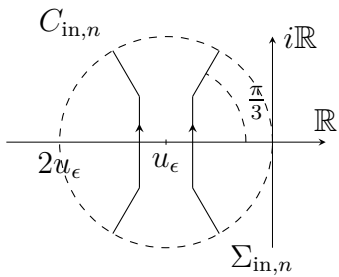

For the rest of the proof, we set $q_{\e} = \left|\frac{G_{\mu}''(u_\e)}{6}\right|$ and we define the two contours $\Sigma_{\textrm{in},n}$ and $C_{\textrm{in},n}$ as in Figure \ref{fig:airysaddlepointcontours} by
\begin{multline}
\Sigma_{\textrm{in},n} = \left\{ u_{\e} + \frac{1}{4\sqrt{3} n^{\frac{1}{3}} q_{\e}^{\frac{1}{3}}} + it ~:~ t \in \left[-\frac{1}{4n^{\frac{1}{3}} q_{\e}^{\frac{1}{3}}}, \frac{1}{4n^{\frac{1}{3}} q_{\e}^{\frac{1}{3}}} \right] \right\} \\
 \bigcup\left\{ u_{\e} + t e^{i\pi/3} ~:~ t \in \left[ \frac{1}{2\sqrt{3}n^{\frac{1}{3}} q_{\e}^{\frac{1}{3}}}, |u_{\e}| \right] \right\} 
\bigcup \left\{ u_{\e} + t e^{-i\pi/3} ~:~ t \in \left[ \frac{1}{2\sqrt{3}n^{\frac{1}{3}} q_{\e}^{\frac{1}{3}}}, |u_{\e}| \right] \right\},
\end{multline}
$C_{\textrm{in}, n}$ being the mirror image of $\Sigma_{\textrm{in}, n}$ with respect to the vertical line passing through $u_{\e}$. Note that the condition $\e n^{\frac{3}{2}} \rightarrow +\infty$ as $n \rightarrow +\infty$ is needed for the path $\Sigma_{\textrm{in},n}$ not to cross the positive real axis.

We will show that the leading behaviour of $K_n^S$ in \eqref{eq:sumkernelalternative} as $n\to+\infty$ comes from the integration over these two local paths, and that it converges to the Airy kernel.

\bigskip

Define
\begin{equation}
K_n^{\textrm{in}}(x,y) = -\frac{n}{(2\pi i)^2 \e^2} \int_{C_{\textrm{in}, n}} ds \int_{\Sigma_{\textrm{in},n}} d\zeta \frac{1}{s-\zeta} \begin{pmatrix} 0 & 1 \end{pmatrix} Y_n^{-1}(s) Y_n(\zeta) \begin{pmatrix} 0 \\ 1 \end{pmatrix} e^{\frac{n}{2\e^2} \left( (x-s)^2 - (y-\zeta)^2 \right) }.
\end{equation}
We first note that the asymptotic expansion \eqref{asymptotics U A} holds in fact not only in the outer region $A_\delta$, but also when $z$ approaches $0$ at a sufficiently slow rate as $n\to+\infty$. If $n\to+\infty$ and $z\to 0$ in such a way that $n^{2}z\to\infty$, \eqref{asymptotics U A} remains valid, but with a weaker error term $o(1)$ instead of $\Or(1/n)$. This follows after a comparison of \eqref{asymptotics U A} and \eqref{as U 0}, by the asymptotic behaviour of the Bessel function $J_\alpha(z)$ as $z\to\infty$.
Since $\e n^{\frac{3}{2}} \rightarrow +\infty$ as $n\rightarrow +\infty$, this means that we can use \eqref{asymptotics U A}, with weaker error term, for all points of the contour $C_{\textrm{in},n}\times \Sigma_{\textrm{in},n}$. We obtain
\begin{equation*}
K_n^{\textrm{in}}(x,y) = -\frac{n}{(2\pi i)^2 \e^2} \int_{C_{\textrm{in}, n}} ds \int_{\Sigma_{\textrm{in},n}} d\zeta e^{\frac{n}{2\e^2} \left( (x-s)^2 - (y-\zeta)^2 \right) } e^{n \left( g_\mu(s) -  g_\mu(\zeta) \right)} f(s,\zeta) \left(1 + o(1)\right),
\end{equation*}
as $n \rightarrow +\infty$. Here $f$ is a certain function independent of $n$ which can be expressed in terms of $\varphi$, which satisfies $f(z,z)=1$ since $\det Y_n = 1$, but whose precise form is not important for us. 
We define
\begin{equation}
\label{eq:defcepsilon}
c_{\e} = \e^{-2}q_{\e}^{-\frac{1}{3}},\qquad \varphi_{\e,n}(z) = \frac{n^{1/3}z}{c_{\e}\e^2} (u_{\e}-a_{\e}).
\end{equation}

 We now let $x$ and $y$ approach the left edge $a_\e$ at an appropriate speed and we  get
\begin{multline}
\frac{1}{c_{\e} n^{\frac{2}{3}}} K_n^{\textrm{in}}\left( a_{\e} - \frac{x}{c_{\e} n^{\frac{2}{3}}}, a_{\e} - \frac{y}{c_{\e} n^{\frac{2}{3}}} \right) = -\frac{n^{\frac{1}{3}}}{c_{\e} (2\pi i)^2 \e^2} e^{\varphi_{\e,n}(x) - \varphi_{\e,n}(y)} \\ \times\ \int_{C_{\textrm{in}, n}} ds \int_{\Sigma_{\textrm{in},n}} d\zeta 
\frac{e^{n \Psi_{\e}(s) + \frac{x n^{\frac{1}{3}}}{\e^2c_{\e}} (s-u_{\e})}}{e^{n \Psi_{\e}(\zeta) + \frac{y n^{\frac{1}{3}}}{\e^2c_{\e}} (\zeta-u_{\e})}} \frac{1}{s-\zeta} f(s,\zeta) \left(1 + o(1)\right),\qquad n\to+\infty.
\end{multline}
Using the asymptotics for the derivatives (\ref{eq:phfctderivatives}) and the fact that $f(z,z) = 1$, we may expand the phase functions and $f$ around $u_{\e}$ to get
\begin{multline}
\frac{1}{c_{\e} n^{\frac{2}{3}}} K_n^{\textrm{in}}\left( a_{\e} - \frac{x}{c_{\e} n^{\frac{2}{3}}}, a_{\e} - \frac{y}{c_{\e} n^{\frac{2}{3}}} \right) = -\frac{n^{\frac{1}{3}}}{c_{\e} (2\pi i)^2 \e^2} e^{\varphi_{\e,n}(x) - \varphi_{\e,n}(y)}\\ \times \ \int_{C_{\textrm{in}, n}} ds \int_{\Sigma_{\textrm{in},n}} d\zeta 
\frac{e^{-nq_{\e} (s-u_{\e})^3 + n^{\frac{1}{3}} q_{\e}^{\frac{1}{3}} x (s-u_{\e}) + r_4^{n,\e}(s)}}{ e^{-nq_{\e} (\zeta-u_{\e})^3 + n^{\frac{1}{3}} q_{\e}^{\frac{1}{3}} y (\zeta-u_{\e}) + r_4^{n,\e}(\zeta)} } \frac{1}{s-\zeta} \left(1 + o(1)\right),
\end{multline}
as $n\to+\infty$, since $s-\zeta\to 0$, where
\begin{equation}
r_4^{n,\e}(z) = \frac{n}{2\e^2} \sum_{j = 4}^{+\infty} \frac{\Psi_{\e}^{(j)}(u_{\e})}{j!} (z-u_{e})^j.
\end{equation}
We now apply the changes of variables $s\mapsto u$ and $\zeta\mapsto v$ defined implicitly by
\begin{align*}
&-nq_{\e} (s-u_{\e})^3  + r_4^{n,\e}(s) = -u^3 ,\\
&-nq_{\e} (\zeta-u_{\e})^3  + r_4^{n,\e}(\zeta) = -v^3, 
\end{align*}
such that $u\sim q_\e^{1/3}n^{1/3}(s-u_\e)$ and $v\sim q_\e^{1/3}n^{1/3}(\zeta-u_\e)$ as $n\to+\infty$.
We then get
\begin{multline}\label{eq:Airylast}
e^{-\varphi_{\e,n}(x) + \varphi_{\e,n}(y)}\frac{1}{c_{\e} n^{\frac{2}{3}}} K_n^{\textrm{in}}\left( a_{\e} - \frac{x}{c_{\e} n^{\frac{2}{3}}}, a_{\e} - \frac{y}{c_{\e} n^{\frac{2}{3}}} \right) = \\
- \frac{1}{(2\pi i)^2} \int_{C_n} du \int_{\Sigma_n} dv \frac{e^{-u^3 + xu}}{e^{-v^3 + yv}} \frac{1}{u-v} (1 + o(1)) = \mathbb K^{\Ai}(x,y)+o(1), \mbox{ as } n \rightarrow +\infty,
\end{multline}
$C_n$ and $\Gamma_n$ being contours that grow to the contours of the Airy integral formula as $n\to+\infty$.

\bigskip

\begin{figure}
\centering
\begin{tikzpicture}[>= stealth, scale=.5]

\draw[very thin] (-3, 0) -- (6,0);
\draw[very thin] (2, -4) -- (2,4);

\draw (60:1cm) -- (60:2cm);
\draw[-<- = .7] (300:1cm) -- (60:1cm);
\draw (300:2cm) -- (300:1cm);
\draw[->- = .5] (300:2cm) -- ++(4,0);
\draw[-<- = .5] (60:2cm) -- node [anchor = south] {$\Sigma$} ++(4,0);

\draw (120:1cm) -- (120:2cm);
\draw[->- = .7] (240:1cm) -- (120:1cm);
\draw (240:2cm) -- (240:1cm);
\draw[->- = .5] (120:2cm) -- ++(0,2);
\draw[-<- = .5] (240:2cm) -- node [anchor = east] {$C$} ++(0,-2);

\draw[very thin] (0,-.05) -- (0,.05);
\draw (0,0) node [anchor = north] {$u_{\e}$};


\draw[dashed, thin] (0,0) circle (2cm);
\draw (-2,0) node[anchor = north] {$2u_{\e}$};

\end{tikzpicture}
\caption{Choices for the contours $C$ and $\Sigma$.}
\label{fig:chosencontours}
\end{figure}
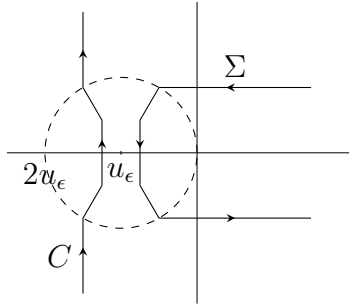

Let $\gamma_1$ and $\gamma_2$ be the contours as in Lemma \ref{lem:steepestdescentcontours}. We take the contours $C$ and $\Sigma$, depending on $n$, to be $C = \bar{\gamma}_2 \cup C_{\textrm{in},n} \cup \gamma_2 \textrm{ and } \Sigma = \bar{\gamma}_1 \cup \Sigma_{\textrm{in},n} \cup \gamma_1$ as pictured in Figure \ref{fig:chosencontours}.
We have
\begin{equation}
\left| \int_{C_{\textrm{in},n}} ds \int_{\gamma_1} d\zeta \frac{f(s,\zeta)}{s-\zeta} \frac{e^{n \Psi_{\e}(s) + n^{\frac{1}{3}} q_{\e}^{\frac{1}{3}} x (s-u_{\e})}}{ e^{n \Psi_{\e}(\zeta) + n^{\frac{1}{3}} q_{\e}^{\frac{1}{3}} y (\zeta-u_{\e})}} \right| \leq \int_{C_{\textrm{in},n}} |ds| \int_{\gamma_1} |d\zeta| \left| \frac{f(s,\zeta)}{s-\zeta} \right| \frac{e^{n {\rm Re}\,(\Psi_{\e}(s)) + n^{\frac{1}{3}} q_{\e}^{\frac{1}{3}} {\rm Re}\,(x(s-u_{\e}))}}{ e^{n {\rm Re}\,(\Psi_{\e}(\zeta)) + n^{\frac{1}{3}} q_{\e}^{\frac{1}{3}} {\rm Re}\,(y(\zeta-u_{\e}))}}.
\end{equation}

In view of (\ref{eq:phfctderivatives}), we know that for $n$ sufficiently large, the real part of the phase function is increasing on $\Sigma_{\textrm{in},n} \cap \{{\rm Im}\, z > 0\}$ and decreasing on $C_{\textrm{in},n}\cap \{{\rm Im}\, z > 0\}$. Moreover, $n^{\frac{1}{3}} q_{\e}^{\frac{1}{3}}$ becomes small compared to $n {\rm Re}\,(\Psi_{\e}(s))$. The function ${\rm Re}\,(\Psi_{\e}(\gamma_1(t)))$ is increasing on $(0,+\infty)$ because of Lemma \ref{lem:steepestdescentcontours}, and moreover grows in such a way that ${\rm Re}\,(\Psi_{\e}(\gamma_1(t))) \sim t^2$ as $t \rightarrow +\infty$, which means that the integral over $\gamma_1$ is exponentially small. In a similar way one shows that all parts of the integration contour $\left(C\times \Sigma\right)\setminus(C_{{\rm in}, n}\times  \Sigma_{{\rm in}, n})$ give exponentially small contributions.

The limit \eqref{eq:AiryLUE} now follows from \eqref{eq:Airylast}.
\end{proof}

\subsection{Proofs of Theorem \ref{thm: ginibre}, Theorem \ref{thm: trunc}, and Theorem \ref{thm:MB}}
\label{sec:prodginibreproofs}

To prove the results, once more we only need to verify that the conditions of Lemma \ref{thm:mainthm} are satisfied.

We start with the Ginibre case, Theorem \ref{thm: ginibre}. Let $K_n$ be the correlation kernel for the squared singular values of a product of Ginibre matrices, given by \eqref{eq:prodginibrekernel}, and let $\widetilde K_n$ be given by \eqref{eq:prodginibrekernelrescaled}.

We will first prove
that there exists $c_1>0$ independent of $n$ such that
\begin{equation}\label{estimate Ginibre}
\left| \frac{1}{n^{m+1}} \widetilde{K}_n\left( \frac{x}{n^{m+1}}, \frac{y}{n^{m+1}} \right) \right| \leq c_1 y^{-1/2}e^{|x|} \mbox{ on } i\R \times \R^+.
\end{equation} This implies condition \eqref{eq:growth2} with $\beta=1/2$.
Using the fact that
\begin{equation}
\frac{\G(t-n+1)}{\G(s-n+1)} = \frac{\G(n-s)}{\G(n-t)} \frac{\sin \pi s}{\sin \pi t}
\end{equation}
(which follows from the reflection formula of the $\G$ function), we obtain
\begin{equation}\label{Knprodgin}
\frac{1}{n^{m+1}} \widetilde{K}_n\left( \frac{x}{n^{m+1}}, \frac{y}{n^{m+1}} \right) = \frac{1}{(2\pi i)^2} \int_{-\frac{1}{2}+i\R} ds \int_{\Sigma_n} dt \prod_{j=0}^m \frac{\G(s+\nu_j+1)}{\G(t+\nu_j+1)} \frac{\G(n-s)}{\G(n-t)} \frac{\sin \pi s}{\sin \pi t} \frac{x^t y^{-s-1}}{s-t} n^{s-t}.
\end{equation}
We can compute the $t$-integral using the residue theorem. The only poles of the $t$-integrand are the ones of $1/ \sin \pi t$ since $1/\G$ is entire, and the $s$- and $t$-contours do not intersect. We obtain
\begin{multline}
\frac{1}{n^{m+1}} \widetilde{K}_n\left( \frac{x}{n^{m+1}}, \frac{y}{n^{m+1}} \right)=\frac{1}{2\pi i}\sum_{k=0}^{n-1} (-1)^k\frac{n^{-k}x^k}{\G(n-k)} \prod_{j=0}^m \frac{1}{\G(k+\nu_j+1)}  \\
\times \int_{-\frac{1}{2}+i\R} \prod_{j=0}^m \G(s+\nu_j+1) \G(n-s) (\sin \pi s) y^{-s-1} n^s \frac{1}{s-k} ds.
\end{multline}
The integral is absolutely convergent since $\left| \sin \pi \left( \frac{1}{2} + it \right) \right| \sim \frac{e^{\pi |t|}}{2}$ and \cite[Formula 5.11.9]{NIST}
\begin{equation}\label{Gamma approx}
\left|\Gamma\left( \frac{1}{2} + \nu_j + it \right) \right| \sim \sqrt{2\pi} |t|^{\nu_j} e^{-\frac{\pi |t|}{2}},
\end{equation}
as $t \rightarrow \pm \infty$. For $s$ on the integration contour, we have $|y^{-s-1} n^{s}| = (ny)^{-\frac{1}{2}}$. Moreover, we have \cite[Inequality 5.6.6]{NIST} $|\G(x+iy)| \leq |\G(x)|$ and thus
\begin{multline}
\left| \frac{1}{n^{m+1}} \widetilde{K}_n\left( \frac{x}{n^{m+1}}, \frac{y}{n^{m+1}} \right) \right| \leq \frac{1}{\pi} y^{-\frac{1}{2}} \int_{-\frac{1}{2}+i\R} \prod_{j=0}^m |\G(s+\nu_j+1)| |\sin \pi s| |ds| \\
\times \sum_{k=0}^{n-1} \frac{|x|^k}{k!} \frac{n^{-k-\frac{1}{2}} \G(n+ 1/2)}{\G(k+1) \G(n-k)}.
\end{multline}
It follows from Stirling's Inequality \cite[Inequality 5.6.1]{NIST} that there exists a constant $C > 0$ such that for $n \in \N$, $k = 0, ..., n-1$,
\begin{equation}
\frac{n^{-k-\frac{1}{2}} \G(n+ 1/2)}{\G(k+1) \G(n-k)} \leq C.
\end{equation}
In $|x|$, we are left with a truncated exponential series. This implies \eqref{estimate Ginibre}.

To prove \eqref{eq:conv2} with $\beta=1/2$, we follow the proof in \cite[Section 5.2]{KuijlaarsZhang}, where \eqref{eq:prodginibrehardedgescaling}
was proved uniformly for $x,y$ in compact subsets of $(0,+\infty)$. We need uniformity of \eqref{eq:conv2} for $u$ in any compact subset of $\mathbb C$ and $v$ in any compact subset of $[0,+\infty)$. Therefore, we fix $R>0$ and assume that $|u|<R$ and $v\in[0,R]$.
Using the asymptotics $\frac{\Gamma(n-s)}{\Gamma(n-t)}=n^{t-s}\left(1+\Or(n^{-1})\right)$ as $n\to\infty$, we obtain from \eqref{Knprodgin} that
\begin{multline}
\frac{v^{1/2}}{n^{m+1}}\widetilde{K}_n\left(\frac{u}{n+1},\frac{v}{n+1}\right)
\\= \frac{1}{(2\pi i)^2} \int_{-\frac{1}{2}+i\R} ds \int_{\Sigma_n} dt \prod_{j=0}^m \frac{\G(s+\nu_j+1)}{\G(t+\nu_j+1)}\frac{\sin\pi s}{\sin\pi t} \frac{u^t v^{-s-\frac{1}{2}} }{s-t}\left(1+\Or(n^{-1})\right),
\end{multline}
as $n\to\infty$.  Using Stirling's formula and the fact that $|\sin\pi t|\geq|\sinh(\pi {\rm Im }\,t)|$, we can bound $\left|\frac{u^t}{\sin\pi t\prod_{j=0}^m\G(t+\nu_j+1)}\right|$ by a function independent of $u$ which decays rapidly as $t\to +\infty$. It follows that we can deform the integration contour $\Sigma_n$ to $\Sigma$, provided $\Sigma$ lies not too close to the real line. For the $s$-integrand, we can use \eqref{Gamma approx} and the fact that $|v^{-s-\frac{1}{2}}|=1$ to bound it uniformly. By the dominated convergence theorem, we can take $\left(1+\Or(n^{-1})\right)$ (with error term independent of $u$ and $v$) out of the integral, and we obtain \eqref{eq:conv2} uniformly in $u$ and $v$.

\begin{remark}\label{remark: Airy Ginibre}
Although the kernel for the squared singular values of products of Ginibre matrices cannot  be expressed in terms of orthogonal polynomials, it can be expressed in terms of the solution to a Riemann-Hilbert problem for multiple orthogonal polynomials \cite[Section 2.2]{KuijlaarsZhang}, and it has a Christoffel-Darboux type formula. If large $n$ asymptotics for this Riemann-Hilbert problem were available, one could be optimistic that similar techniques as in Section \ref{sec:lueproofs2} can be used in the super-critical case to prove convergence to the Airy kernel. 
\end{remark}

\begin{remark}\label{remark: analogy}
The proofs of Theorem \ref{thm: trunc} and Theorem \ref{thm:MB} are very similar to the proof of Theorem \ref{thm: ginibre}. 
For Theorem \ref{thm: trunc}, one has to start with the expression \eqref{eq:KieburgKuijlaarsStivignykernel} for the kernel $K_n$ instead of 
 \eqref{eq:prodginibrekernel}. Then one shows in a similar way as in the Ginibre case that
 \begin{equation}\label{eq:conv trunc}
\lim_{n \rightarrow +\infty} v^{1/2} \frac{1}{c_n} K_n\left( \frac{u}{c_n}, \frac{v}{c_n} \right) = v^{1/2}\K_{\nu,\mu}^{T}(u,v),
\end{equation}
uniformly for $u$ in any compact subset of $\mathbb C$ and $v$ in any compact subset of $[0,+\infty)$, and that there are constants $C_1, C_2$ such that
\begin{equation}\label{eq:growth trunc}
|K_n(u,v)| \leq C_1v^{-1/2}c_n^{1/2}   e^{C_2 c_n (|u|+|v|)},
\end{equation}
for every $(u,v) \in i\R \times [0,+\infty)$ and $n > n_0$. The two above conditions are almost the same as the assumptions in Lemma \ref{thm:mainthm}, except for the fact that $c_n$ now plays the role of $cn^\gamma$. From the proof of Lemma \ref{thm:mainthm}, it is straightforward to check that conditions \eqref{eq:conv trunc} and \eqref{eq:growth trunc} imply
\eqref{eq:scaling sum subcrit} and \eqref{eq:sumconv} with $cn^\gamma$ replaced by $c_n$. This leads to the proof of Theorem \ref{thm: trunc} without further complications

For Theorem \ref{thm:MB}, one has to use the expression \eqref{kernel MB} or \eqref{kernel MB2} instead. The rest of the proof is a straightforward adaptation of the proof of Theorem \ref{thm: ginibre}, and we do not give the details here.
\end{remark}

\section*{Acknowledgements}
The authors are supported by the European Research Council under the European Union's Seventh Framework Programme (FP/2007/2013)/ ERC Grant Agreement 307074 and by the Belgian Interuniversity Attraction Pole P07/18.


\begin{thebibliography}{99}
\addcontentsline{toc}{section}{References}

\bibitem{AdlervanMoerbekeWang} M. Adler, P. van Moerbeke, and D. Wang, Random matrix minor processes related to percolation theory, {\em Random Matrices Theory Appl.} {\bf 2} (2013), no. 4, 1350008, 72 pp.

\bibitem{Akemann-Ipsen15}
G.~Akemann and J.~R. Ipsen.
\newblock Recent exact and asymptotic results for products of independent
  random matrices, {\em Acta Physica Polonica B} {\bf 46} (2015), no. 9, 1747--1784.

\bibitem{Akemann-Ipsen-Kieburg13}
G.~Akemann, J.~R. Ipsen, and M.~Kieburg.
\newblock Products of rectangular random matrices: {S}ingular values and
  progressive scattering,
\newblock {\em Phys. Rev. E} {\bf 88} (2013), 052118.

\bibitem{Akemann-Kieburg-Wei13}
G.~Akemann, M.~Kieburg, and L.~Wei.
\newblock Singular value correlation functions for products of {W}ishart random
  matrices,
  \newblock {\em J. Phys. A} {\bf 46}    
   (2013), no. 27, 275205, 22 pp.

\bibitem{AtkinClaeysMezzadri} M.R. Atkin, T. Claeys, and F. Mezzadri, Random matrix ensembles with singularities and a hierarchy of Painlev\'e III equations, {\em Int. Math. Res. Notices} {\bf 2015} (2015), doi: 10.1093/imrn/rnv195, 56 pages. 


\bibitem{BertolaEynardHarnad}
M. Bertola, B. Eynard, and J. Harnad, Duality, biorthogonal polynomials and multimatrix
models, {\em Comm. Math. Phys.} {\bf 229} (2002), 73--120.


\bibitem{BGS}
M. Bertola, M. Gekhtman, and J. Szmigielski, Cauchy-Laguerre two-matrix model and the Meijer-$G$ random point field, {\em Commun. Math. Phys.} {\bf 326}   
 (2014), no. 1, 111--144.

\bibitem{Biane}
P. Biane, On the free convolution with a semi-circular distribution, \emph{Indiana Univ. Math. J.} \textbf{46} (1997), no. 3,  705--718.

\bibitem{Borodin}
A. Borodin, Biorthogonal ensembles, \emph{Nuclear Phys. B} {\bf 536} (1999), no. 3, 704--732.

\bibitem{BrezinHikami}
E. Br\'ezin and S. Hikami, Universal singularity at the closure of a gap in a random
matrix theory, {\em Phys. Rev. E} {\bf 57} (1998), 4140--4149.



\bibitem{Cheliotis} D. Cheliotis, Triangular random matrices and biorthogonal ensembles, arxiv:1404.4730.


\bibitem{ClaeysKuijlaarsWang}
T. Claeys,  A.B.J. Kuijlaars, and D. Wang, Correlation kernel for sums and products of random matrices, {\em Random Matrices Theory Appl.} {\bf 4} (2015), no. 4, 1550017, 31 pp. 

\bibitem{DeiftZhou} P. Deift and X. Zhou, A steepest descent method for oscillatory Riemann-Hilbert
problems. Asymptotics for the MKdV equation, {\em Ann. Math.} {\bf 137} (1993), 295--368.

\bibitem{Duits}
M. Duits, Painlev\'e kernels in {H}ermitian matrix models, {\em Constr. Approx.} {\bf 39} (1) (2014), 173--196.

\bibitem{FokasItsKitaev}
A.S. Fokas, A.R. Its, and A.V. Kitaev, The isomonodromy approach to matrix models in 2D quantum gravity, \emph{Commun. Math. Phys.} \textbf{147} (1992), no. 2, 395--430.

\bibitem{Forrester}
P.J. Forrester, Eigenvalue statistics for product complex Wishart matrices, \emph{J. Phys. A: Math. Theor.} \textbf{47} (2014), no. 34, 345202, 22 pp.

\bibitem{ForresterGrela}
P.J. Forrester and J. Grela, Hydrodynamical spectral evolution for random matrices, arxiv:1507.07274.

\bibitem{ForresterLiu}
P.J. Forrester and D.-Z. Liu, Singular values for products of complex Ginibre matrices with a source: hard edge limit and phase transition, arxiv:1503.07955.

\bibitem{ForresterWang}
P.J. Forrester and D. Wang, Muttalib--Borodin ensembles in random matrix theory --- realisations and correlation functions, arXiv:1502.07147.

\bibitem{Hardy}
A. Hardy, Average characteristic polynomials of determinantal point processes, \emph{Ann. Inst. H. Poincare Probab. Statist.} \textbf{51} (2015), no. 1, 283--303.

\bibitem{Johansson} K. Johansson, Non-intersecting paths, random tilings and random matrices, {\em Probab. Theory Related Fields} {\bf 123} (2002), no. 2, 225--280.

\bibitem{KieburgKuijlaarsStivigny}
M. Kieburg, A.B.J. Kuijlaars, and D. Stivigny, Singular value statistics of matrix products with truncated unitary matrices, to appear in {\em Int. Math. Res. Notices}, arXiv:1501.03910.



\bibitem{KuijlaarsStivigny}
A.B.J. Kuijlaars and D. Stivigny, Singular values of products of random matrices and polynomial ensembles, \emph{Random Matrices Theory Appl.} \textbf{3} (2014), 1450011, 22 pp.


\bibitem{KuijlaarsZhang}
A.B.J. Kuijlaars and L. Zhang, Singular values of products of Ginibre random matrices, multiple orthogonal polynomials and hard edge scaling limits, \emph{Commun. Math. Phys.} \textbf{332} (2014), 759--781.

\bibitem{LiuWangZhang}
D.Z. Liu, D. Wang, and L. Zhang, Bulk and soft-edge universality for singular values of products of Ginibre random matrices, arXiv:1412.6777.

\bibitem{Muttalib} K.A. Muttalib, Random matrix models with additional interactions, \emph{J. Phys. A} \textbf{28} (1995), no. 5, L159--L164.

\bibitem{Neuschel}
T. Neuschel, Plancherel-Rotach formulae for average characteristic polynomials of products of Ginibre random matrices and the Fuss-Catalan distribution, {\em Random Matrices Theory Appl.} {\bf 03} (2013), no. 1, 1450003, 18 pp. 

\bibitem{Nica-Speicher06}
A.~Nica and R.~Speicher.
\newblock {\em Lectures on the Combinatorics of Free Probability}, volume 335
  of {\em London Mathematical Society Lecture Note Series}.
\newblock Cambridge University Press, Cambridge, 2006.

\bibitem{NIST}
F. Olver, D. Lozier, R. Boisvert, C. Clark, editors, ``\,NIST Handbook of Mathematical Functions", Cambridge University Press, Cambridge, 2010.

\bibitem{OlverRao}
S. Olver and N. Raj Rao, Numerical computation of convolutions in free probability theory, arXiv:1203.1958v2.
 
\bibitem{PensonZyckowski}
K.A. Penson and K. Zyckowski, Product of Ginibre matrices: Fuss-Catalan and Raney distributions. \emph{Phys. Rev. E}, 
{\bf 83} (2011), no. 6, 061118.

\bibitem{SaffTotik}
    E.B. Saff and V. Totik,
    ``\,Logarithmic Potentials with External Fields",
    Springer-Verlag, New-York, 1997.

\bibitem{Speicher}
R. Speicher, Free probability theory, in ``\,The Oxford Handbook of Random Matrix Theory", Oxford University Press, 452--470, 2011.

\bibitem{Vanlessen}
M. Vanlessen, Strong asymptotics of Laguerre-type orthogonal polynomials and applications in random matrix theory, \emph{Constr. Approx.} \textbf{25} (2007), no. 2, 125--175.

\bibitem{XDZ} S.-X. Xu, D. Dai, and Y.-Q. Zhao, Critical edge behavior and the Bessel to Airy transition in the singularly perturbed Laguerre unitary ensemble,  {\em Commun. Math. Phys.} {\bf 332} (2014), no. 3, 1257--1296. 

\bibitem{Zyczkowski-Sommers00}
K.~{\.Z}yczkowski and H.-J. Sommers.
\newblock Truncations of random unitary matrices,
\newblock {\em J. Phys. A}, {\bf 33} (2000), no. 10, 2045--2057.

\bibitem{Zhang} L. Zhang, Local universality in biorthogonal Laguerre ensembles, to appear in {\em J. Stat. Phys.}, arxiv:1502.03160.


\bibitem{ZinnJustin}
P. Zinn-Justin, Universality of correlation functions of Hermitian random matrices
in an external field, {\em Comm. Math. Phys.} {\bf 194} (1998), 631--650.


\end{thebibliography}
\end{document}